\documentclass[11pt]{article}

\usepackage[paper=letterpaper, margin=1in]{geometry}
\usepackage{ifthen}

\ifthenelse{\isundefined{\GenerateShortVersion}}
{
\def\LongVersion{}
\def\LongVersionEnd{}
\long\def\ShortVersion#1\ShortVersionEnd{}
}
{
\def\ShortVersion{}
\def\ShortVersionEnd{}
\long\def\LongVersion#1\LongVersionEnd{}
}

\newcommand{\Ignore}[1]{\ignorespaces}

\usepackage{amsmath}
\usepackage{amssymb}

\usepackage{amsthm}

\usepackage{ifpdf}

\usepackage{authblk}

\ifpdf
\usepackage[pdftex]{graphicx}
\else
\usepackage[dvips]{graphicx}
\fi

\usepackage[%
  ocgcolorlinks,%
  linkcolor=blue,%
  filecolor=blue,%
  citecolor=blue,%
  urlcolor=blue]{hyperref}

\usepackage{enumitem}

\usepackage{multirow}
\usepackage[singlelinecheck=false, labelfont=bf]{caption}

\usepackage{algorithm}
\usepackage{algpseudocode}
\algtext*{EndWhile} % Removes the "end while" text
\algtext*{EndIf} % Removes the "end if" text
\algtext*{EndFor} % Removes the "end for" text

\makeatletter
\renewcommand*{\ALG@name}{Pseudocode}
\makeatother

\LongVersion %{
\LongVersionEnd %}

\ShortVersion %{
\usepackage{times}
\ShortVersionEnd %}

\ShortVersion %{
\renewcommand{\paragraph}[1]{\par\noindent\textbf{#1}}
\ShortVersionEnd %}

\newtheorem{theorem}{Theorem}[section]
\newtheorem{lemma}[theorem]{Lemma}

\newtheorem{corollary}[theorem]{Corollary}

\newtheorem{claim}[theorem]{Claim}

\newtheorem{GlobalTheorem}{Theorem}

\newtheorem*{corollary*}{Corollary}

\theoremstyle{definition}
\newtheorem{definition}[theorem]{Definition}

\newtheorem{remark}[theorem]{Remark}

\newtheorem*{definition*}{Definition}
\theoremstyle{plain}

\newenvironment{DenseItemize}[0]
{\begin{itemize}[nosep, leftmargin=*]}
{\end{itemize}}

\LongVersion %{
\newenvironment{MathMaybe}[0]
{\begin{displaymath}\ignorespaces}
{\end{displaymath}\ignorespacesafterend}
\LongVersionEnd %}
\ShortVersion %{
\newenvironment{MathMaybe}[0]
{\begin{math}\ignorespaces}
{\end{math}}
\ShortVersionEnd %}

\LongVersion %{
\newcommand{\Left}{\left}
\newcommand{\Right}{\right}
\newcommand{\LLeft}{\left}
\newcommand{\RRight}{\right}
\LongVersionEnd %}
\ShortVersion %{
\newcommand{\Left}{\big}
\newcommand{\Right}{\big}
\newcommand{\LLeft}{\Big}
\newcommand{\RRight}{\Big}
\ShortVersionEnd %}

\newenvironment{subproof}[1][\proofname]{%
\begin{proof}[#1]%
}{%
\end{proof}%
}

\def\qedlabel#1{\def\theqedlabel{#1}}

\newcommand{\Sect}{Sec.}
\newcommand{\Thm}{Thm.}
\newcommand{\Lem}{Lem.}

\begin{document}

\title{Approximating Generalized Network Design under (Dis)economies of Scale
with Applications to Energy Efficiency\footnote{An extended abstract of this paper is to appear in the 50th Annual ACM Symposium on the Theory of Computing (STOC 2018).}%\\%
\LongVersion %{
%(Full Version)
\LongVersionEnd %}
\ShortVersion %{
(Extended Abstract)
\ShortVersionEnd %}
}

\author{Yuval Emek\footnote{The work of Yuval Emek was supported in part by an Israeli Science Foundation grant number 1016/17.} , Shay Kutten\footnote{The work of Shay Kutten was supported in part by a grant from the ministry of science in the program that is joint with JSPS and in part by the BSF.} , Ron Lavi\footnote{The work of Ron Lavi was supported by a grant from the MINERVA foundation (ARCHES 2011) and by the ISF-NSFC joint research program (grant No.~2560/17).} , Yangguang Shi\footnote{The work of Yangguang Shi was partially supported at the Technion by a fellowship of the Israel Council for Higher Education.
}}
\affil{Technion - Israel Institute of Technology.\\
\texttt{\{yemek, kutten, ronlavi, shiyangguang\}@ie.technion.ac.il}}

\date{}

%\begin{titlepage}

\maketitle

\begin{abstract}
In a \emph{generalized network design (GND)} problem, a set of
\emph{resources} are assigned (non-exclusively) to multiple 
\emph{requests}.
Each request contributes its weight to the resources it uses and the
total \emph{load} on a resource is then translated to the cost it incurs via a
a resource specific cost function.
Motivated by \emph{energy efficiency} applications, recently, there
is a growing interest in GND using cost functions that exhibit
\emph{(dis)economies of scale ((D)oS)}, namely, cost functions that appear
subadditive for small loads and superadditive for larger loads.

The current paper advances the existing literature on approximation algorithms
for GND problems with (D)oS cost functions in various aspects:
(1)
while the existing results are restricted to \emph{routing} requests in
undirected graphs, identifying the resources with the graph's edges,
the current paper presents a generic \emph{approximation framework} that
yields approximation results for a much wider family of requests (including
various types of \emph{Steiner tree} and \emph{Steiner forest} requests) in
both directed and undirected graphs, where the resources can be identified
with either the edges or the vertices;
(2)
while the existing results assume that a request contributes the same
weight to each resource it uses, our approximation framework allows for
\emph{unrelated} weights, thus providing the first non-trivial approximation
for the problem of \emph{scheduling unrelated parallel machines} with
(D)oS cost functions;
(3)
while most of the existing approximation algorithms are based on convex
programming, our approximation framework is fully \emph{combinatorial} and
runs in strongly polynomial time;
(4)
the family of (D)oS cost functions considered in the current paper is more
general than the one considered in the existing literature, providing a more
accurate abstraction for practical energy conservation scenarios;
and
(5)
we obtain the first approximation ratio for GND with (D)oS cost functions that
depends only on the parameters of the resources' technology and does not grow
with the number of resources, the number of requests, or their weights.
The design of our approximation framework relies heavily on Roughgarden's
\emph{smoothness} toolbox (JACM 2015), thus demonstrating the possible
usefulness of this toolbox in the area of approximation algorithms.
\end{abstract}

\noindent
\textbf{Keywords:}
Approximation algorithms,
generalized network design,
(dis)economies of scale,
energy consumption,
real exponent polynomial cost functions,
smoothness,
best response dynamics

\ShortVersion %{
\vskip 5mm

\begin{center}
\Large{%
A full version that contains all missing proofs along with some additional
material is attached to this 10 page extended abstract.\footnote{%
For ease of reference, theorem numbers in this extended abstract correspond to
their counterparts in the full version.}
}
\end{center}
\ShortVersionEnd %}

\LongVersion
\vskip 5mm
\LongVersionEnd

%\renewcommand{\thepage}{}
%\end{titlepage}

%\clearpage

%\pagenumbering{arabic}

%%%%%%%%%%%%%%%%%%%%%%%%%%%%%%%%%%%%%%%%%%%%%%%%%%%%%%%%%%%%%%%%%%%%%%%%%%%%%%
\section{Introduction}
\label{sec:introduction}
%%%%%%%%%%%%%%%%%%%%%%%%%%%%%%%%%%%%%%%%%%%%%%%%%%%%%%%%%%%%%%%%%%%%%%%%%%%%%%

%%%%%%%%%%%%%%%%%%%%%%%%%%%%%%%%%%%%%%%
\paragraph{Generalized Network Design.}
%%%%%%%%%%%%%%%%%%%%%%%%%%%%%%%%%%%%%%%
An instance $\mathcal{I}$ of a \emph{generalized network design (GND)}
problem is defined over a finite set $E$ of \emph{resources} and $N$ abstract
\emph{requests}.
Each request
$i \in [N]$
is served by choosing some \emph{reply}
$p_{i} \subseteq E$
from request $i$'s \emph{reply collection}
$P_{i} \subseteq 2^{E}$.
Serving request $i$ with reply $p_{i}$ contributes $w_{i}(e)$ units to the
\emph{load} $l_{e}$ on resource $e$ for each
$e \in p_{i}$,
where
$w_{i} \in \mathbb{Z}_{\geq 1}^{E}$
is the \emph{weight vector} associated with request $i$ (specified in
$\mathcal{I}$).
We emphasize that our GND setting supports \emph{unrelated} weights,
that is, request $i$ may contribute different weights to the load on different
resources in $p_{i}$.

One should serve all the requests of the instance $\mathcal{I}$ with replies
$p = \{ p_{i} \}_{i \in [N]}$,
satisfying
$p_{i} \in P_{i}$
for every
$i \in [N]$,
under the objective of minimizing the \emph{total cost} $C(p)$.
This is defined as
$C(p) = \sum_{e \in E} F_{e}(l_{e})$,
where
$F_{e} : \mathbb{Z}_{\geq 0} \rightarrow \mathbb{R}_{\geq 0}$
is a resource \emph{cost function} that maps the load
$l_{e} = l_{e}^{p} = \sum_{i \in [N] : e \in p_{i}} w_{i}(e)$
induced by $p$ on resource $e$ to the cost incurred by that resource.

We restrict our attention to GND problems with \emph{succinctly represented}
requests, namely, requests whose reply collections $P_{i}$ can be specified
using $\operatorname{poly}(|E|)$ bits.
These requests are often defined by identifying the resource set $E$ with the
edge set of a (directed or undirected) graph
$G = (V, E)$, giving rise to, e.g., the following request types:
\begin{DenseItemize}

\item
\emph{routing} requests in directed or undirected graphs, where given a pair
$(s_{i}, t_{i}) \in V \times V$
of \emph{terminals}, the reply collection $P_{i}$ consists of all
$(s_{i}, t_{i})$-paths
in $G$;

\item
\emph{multi-routing} requests in directed or undirected graphs, where
given a collection
$D_{i} \subseteq V \times V$
of \emph{terminal} pairs, the reply collection $P_{i}$ consists of all
edge subsets
$F \subseteq E$
such that the subgraph $(V, F)$ admits an $(s, t)$-path for every
$(s, t) \in D_{i}$
(useful for designing a multicast scheme);\footnote{%
Notice that the multi-routing request given by $D_{i}$ cannot be
(trivially) reduced to $|D_{i}|$ (single-)routing requests since a reply $F$
for the former contributes $w_{i}(e)$ units to the load on edge
$e \in F$
``only once'', even if this edge is used to connect multiple terminal pairs in
$D_{i}$.
}
and

\item
\emph{set connectivity} (resp., \emph{set strong connectivity}) requests in
undirected (resp., directed) graphs, where
given a set
$T_{i} \subseteq V$
of \emph{terminals}, the reply collection $P_{i}$ consists of all edge
subsets that induce on $G$ a connected (resp., strongly connected) subgraph
that spans $T_{i}$
(useful for designing an \emph{overlay network}).

\end{DenseItemize}
Alternatively, one can identify the resource set $E$ with the vertex set of a
graph, obtaining the vertex variants of the aforementioned request types, or
with any other combinatorial structure as long as it fits into the
aforementioned setting.

%%%%%%%%%%%%%%%%%%%%%%%%%%%%%%%%%%%%%%%
\paragraph{(Dis)economies of Scale.}
%%%%%%%%%%%%%%%%%%%%%%%%%%%%%%%%%%%%%%%
The classic network design literature addresses scenarios where the higher
the load on a resource is, the lower is the cost per unit load, thus making it
advisable to share network resources among requests, commonly known as
\emph{buy-at-bulk} network design
\cite{Awerbuch1997BAB, Andrews2004, Charikar2005NUM, Chekuri2010AAN}.
More formally, the cost functions $F_{e}(\cdot)$ in buy-at-bulk network
design are assumed to be  \emph{subadditive}, i.e., they exhibit
\emph{economies of scale}.
Recently, there is a growing interest in investigating network design problems
with \emph{superadditive} cost functions (i.e., cost functions exhibiting
\emph{diseconomies of scale})
\cite{Andrews2012RPM, Makarychev2014SOP}
or even cost functions that may appear subadditive for small loads and
superadditive for larger loads
\cite{Andrews2012RPM, Antoniadis2014HHE, Andrews2016MCN}, referred to as cost
functions exhibiting \emph{(dis)economies of scale ((D)oS)}
\cite{Andrews2016MCN}.

The (D)oS cost functions studied so far in the context of network design
capture the \emph{energy consumption} of network devices employing the popular
\emph{speed scaling} technique
\cite{YaoDS1995, IraniP2005, BansalKP2007, Nedevschi2008RNE,
Christensen2010IEEE, Albers2010, Andrews2013RSE, Andrews2016MCN}
that allows the device to adapt its power level to its actual load.
Given a global constant parameter
$\alpha \in \mathbb{R}_{> 1}$
(a.k.a.\ the \emph{load exponent}),
an energy consumption cost function for resource
$e \in E$
is defined by setting
\begin{equation} \label{equation:energy-consumption-cost-function}
F_{e}(l_{e})
\, = \,
\begin{cases}
0 \, , & l_{e} = 0 \\
\sigma_{e} + \xi_{e} \cdot l_{e}^{\alpha} \, , & l_{e} > 0
\end{cases} \, ,
\end{equation}
where
$\sigma_{e} \in \mathbb{R}_{\geq 0}$
(the \emph{startup cost})
and
$\xi_{e} \in \mathbb{R}_{> 0}$
(the \emph{speed scaling factor})
are parameters of $e$.

\LongVersion %{
This paper improves
\LongVersionEnd %}
\ShortVersion %{
We improve
\ShortVersionEnd %}
the existing results on approximation algorithms for GND
with energy consumption cost functions in various aspects (see
\Sect{}~\ref{sec:comp-prev-results}).
In fact, our results apply to a more general class of resource cost functions
exhibiting (D)oS, referred to as \emph{real exponent polynomial (REP)} cost
functions.
Given global constant parameters
$q \in \mathbb{Z}_{\geq 1}$
and
$\alpha_{1}, \dots, \alpha_{q} \in \mathbb{R}_{> 1}$,
a REP cost function for resource
$e \in E$
is defined by setting
\begin{equation} \label{equation:REP-cost-function}
F_{e}(l_{e})
\, = \,
\begin{cases}
0 \, , & l_{e} = 0 \\
\sigma_{e} + \sum_{j \in [q]} \xi_{e, j} \cdot l_{e}^{\alpha_{j}} \, , & l_{e} > 0
\end{cases} \, ,
\end{equation}
where
$\sigma_{e} \in \mathbb{R}_{\geq 0}$
and
$\xi_{e, 1}, \dots, \xi_{e, q} \in \mathbb{R}_{\geq 0}$
are parameters of $e$, constrained by requiring that
$\xi_{e, j} > 0$
for at least one
$j \in [q]$.

On top of the theoretical interest in studying more general cost functions,
there is also a practical motivation behind their investigation.
While some of the theoretical literature on energy efficient network design
considers the special case of
(\ref{equation:energy-consumption-cost-function}) where
$\sigma_e = 0$
(see \Sect{}~\ref{sec:comp-prev-results}), it has been claimed
\cite{Andrews2016MCN, Antoniadis2014HHE} that the startup cost component is
crucial for better capturing practical energy consumption structures.
In fact, in realistic communication networks, even the energy consumption cost
functions of (\ref{equation:energy-consumption-cost-function}) may not be
general enough since a link often consists of several different devices (e.g.,
transmitter/receiver, amplifier, adapter), all of which are operating when the
link is in use.
As their energy consumption may vary in terms of the load exponents and speed
scaling factors, the functions presented in
(\ref{equation:energy-consumption-cost-function}) do not provide a suitable
abstraction for the link's energy consumption and the more general REP cost
functions (\ref{equation:REP-cost-function}) should be employed.

%%%%%%%%%%%%%%%%%%%%%%%%%%%%%%%%%%%%%%%
\paragraph{Approximation Framework.}
%%%%%%%%%%%%%%%%%%%%%%%%%%%%%%%%%%%%%%%
Our main contribution is a novel \emph{approximation framework} for GND
problems with REP resource cost functions.
This framework yields an approximation algorithm when provided access to an
appropriate oracle that we now turn to define.
A \emph{reply $\varrho$-oracle},
$\varrho \geq 1$,
for a family $\mathcal{Q}$ of succinctly represented requests is an efficient
procedure that given
a resource set $E$,
the reply collection
$R \subseteq 2^{E}$
(specified succinctly) of a request in $\mathcal{Q}$,
and a \emph{toll function}
$\tau : E \rightarrow \mathbb{R}_{> 0}$,
returns some reply
$r \in R$
that minimizes the total toll
$\tau(r) = \sum_{e \in r} \tau(e)$
up to factor $\varrho$,
i.e., it satisfies
$\tau(r) \leq \varrho \cdot \tau(r')$
for every
$r' \in R$.
An \emph{exact reply oracle} is a reply $\varrho$-oracle with
$\varrho = 1$.

Notice that the optimization problem behind the reply oracle is \emph{not} a
GND problem:
it deals with a \emph{single} request (rather than multiple requests) and the
role of the resource cost functions (combined with the weight vectors) is now
taken by the (single) toll function.
In particular, while all the (specific) GND problems mentioned in this paper
are intractable (to various extents of inapproximability
\cite{AzarERW2004, Andrews2012RPM, Roughgarden2014}),
the request classes corresponding to some of them admit exact reply oracles.

For example, routing requests (in directed and undirected graphs) admit an
exact reply oracle implemented using, e.g., Dijkstra's shortest path algorithm
\cite{Dijkstra1959, FredmanT1987}.
In contrast,
set connectivity requests in undirected graphs,
set strong connectivity requests in directed graphs,
and multi-routing requests in undirected and directed graphs
do not admit exact reply oracles unless
$\mathrm{P} = \mathrm{NP}$
as these would imply exact (efficient) algorithms for the
\emph{Steiner tree},
\emph{strongly connected Steiner subgraph},
\emph{Steiner forest},
and
\emph{directed Steiner forest}
problems, respectively.
However, employing known approximation algorithms for the latter (Steiner)
problems, one concludes that%
\LongVersion %{
:
set connectivity requests in undirected graphs admit a reply $\varrho$-oracle
for 
$\varrho \leq 1.39$
\cite{Byrka2013STA};
set strong connectivity requests in directed graphs admit a reply
$t^{\epsilon}$-oracle,
where
$t = |T|$
is the number of terminals \cite{CharikarCCDGGL1998};
multi-routing requests in undirected graphs admit a reply $2$-oracle
\cite{AgrawalKR1995};
and
multi-routing requests in directed graphs admit a reply
$k^{1 / 2 + \epsilon}$-oracle,
where
$k = |D|$
is the number of terminal pairs \cite{ChekuriEGS2011}.
This means, in particular, that
\LongVersionEnd %}
\ShortVersion \ \ShortVersionEnd
set connectivity replies and multi-routing replies in undirected graphs always
admit a reply $\varrho$-oracle with a constant approximation ratio $\varrho$,
whereas
set strong connectivity replies and multi-routing replies in directed graphs
admit such an oracle whenever $|T|$ and $|D|$ are fixed%
\ShortVersion %{
\ \cite{AgrawalKR1995, CharikarCCDGGL1998, ChekuriEGS2011, Byrka2013STA}%
\ShortVersionEnd %}
.
The guarantees of our approximation framework are cast in the following
theorem.

\begin{GlobalTheorem} \label{theorem:approximation-framework}
Consider some GND problem $\mathcal{P}$ with succinctly represented requests
using REP resource cost functions as defined in
(\ref{equation:REP-cost-function}).
Suppose that the requests of $\mathcal{P}$ admit a reply $\varrho$-oracle
$\mathcal{O}_{\mathcal{P}}$.
When provided with free access to $\mathcal{O}_{\mathcal{P}}$, our
approximation framework yields a randomized efficient approximation algorithm
$\mathcal{A}_{\mathcal{P}}$ for $\mathcal{P}$ whose approximation ratio is
\begin{MathMaybe}
O \left(
\varrho^{\max_{j} \alpha_{j} + 1}
+
\varrho \cdot \max\nolimits_{e} \min\nolimits_{j}
\left( \frac{\sigma_{e}}{\xi_{e, j}} \right)^{1 / \alpha_{j}}
\right)
\end{MathMaybe}
with high probability.
Moreover, our approximation framework is fully combinatorial and it runs in
strongly polynomial time, so if $\mathcal{O}_{\mathcal{P}}$ is implemented to
run in strongly polynomial time, then $\mathcal{A}_{\mathcal{P}}$ also runs in
strongly polynomial time.
\end{GlobalTheorem}

We emphasize that when
$\varrho = O (1)$,
the approximation ratio promised in
\Thm{}~\ref{theorem:approximation-framework} becomes
\[
O \left(
1 +
\max\nolimits_{e} \min\nolimits_{j}
\left( \frac{\sigma_{e}}{\xi_{e, j}} \right)^{1 / \alpha_{j}}
\right)
\]
which is free of any dependence on the number $|E|$ of resources, the number
$N$ of requests, and the weight vectors
$\{ w_{i} \}_{i \in [N]}$;
rather, it depends only on the parameters ($\sigma_{e}$, $\xi_{e, j}$) of
the network resources' technology (speed scaling in case
$q = 1$).
Notice that the hidden expressions in our $O$ notations may depend on the
parameters $q$ and
$\alpha_{1}, \dots, \alpha_{q}$
assumed to be constants throughout this paper.

%%%%%%%%%%%%%%%%%%%%%%%%%%%%%%%%%%%%%%%
\subsection{Comparison to Existing Results}
\label{sec:comp-prev-results}
%%%%%%%%%%%%%%%%%%%%%%%%%%%%%%%%%%%%%%%

%%%%%%%%%%%%%%%%%%%%%%%%%%%%%%%%%%%%%%%
\paragraph{GND with Routing Requests.}
%%%%%%%%%%%%%%%%%%%%%%%%%%%%%%%%%%%%%%%
The existing literature on (generalized) network design beyond subadditive
resource cost functions
\cite{Andrews2012RPM, Antoniadis2014HHE, Makarychev2014SOP, Andrews2016MCN}
focuses on routing requests, identifying the resources with the edges of a
graph, and with the exception of \cite{Makarychev2014SOP}, it is restricted to
undirected graphs and \emph{related} weights, i.e.,
$w_{i}(e) = w_{i}$
for every
$e \in E$.
In contrast, the current paper handles a wider class of request types over
much more general combinatorial structures (including both directed and
undirected graphs) and our approximation framework supports unrelated
weights.
Moreover, the current paper addresses the general REP cost functions
(\ref{equation:REP-cost-function}), whereas as stated beforehand, the existing
literature addresses only the energy consumption cost functions
(\ref{equation:energy-consumption-cost-function}) and special cases thereof
(Table~\ref{tab:comparison-routing-energy} summarizes the relevant
approximation upper bounds).

Specifically, Makarychev and Sviridenko \cite{Makarychev2014SOP} consider
purely superadditive cost functions by restricting
(\ref{equation:energy-consumption-cost-function}) to
$\sigma_{e} = 0$
for all
$e \in E$,
obtaining an approximation ratio of
$(1 + \epsilon) \mathcal{B}_{\alpha}$,
where $\mathcal{B}_{\alpha}$ is the fractional Bell number with parameter
$\alpha$.
This improves the prior
$O \Left( \log^{\alpha - 1} w_{\max} \Right)$
upper bound of Andrews et al.~\cite{Andrews2012RPM},
where
$w_{\max} = \max_{i \in [N]} w_{i}$.
The case where the startup cost $\sigma_{e}$ may be positive is addressed by
Antoniadis et al.~\cite{Antoniadis2014HHE}, obtaining an approximation ratio of
$O \Left( \log^{\alpha} N \Right)$,
but this result is limited to the uniform case where
$w_{i} = 1$
for all
$i \in [N]$.

As stated in \cite{Andrews2016MCN, Antoniadis2014HHE}, for a more accurate
abstraction of practical energy conservation scenarios, the cost function
definition of (\ref{equation:energy-consumption-cost-function}) with positive
startup costs and arbitrary (related) weights is unavoidable.
In this setting, three different approximation ratios have been devised by
Andrews et al.:
$O \Left(\Left( 1 + \max_{e} \frac{\sigma_{e}}{\xi_{e}} \Right)^{1 / \alpha}
\log^{\alpha - 1} w_{\max} \Right)$
and
$O \Left( N + \log^{\alpha - 1} w_{\max} \Right)$
in \cite{Andrews2012RPM};
and
$\operatorname{polylog}(N) \cdot \log^{\alpha - 1} w_{\max}$
in \cite{Andrews2016MCN}.\footnote{%
Actually, in \cite{Andrews2016MCN}, the startup cost term in the cost function
is somewhat restricted.}

We emphasize that these three approximation ratios grow with
the number $N$ of traffic requests and/or the maximum weight $w_{\max}$,
whereas the approximation ratio established in the current paper depends only
on the parameters of the network resources' technology.
Furthermore, the algorithms behind these approximation ratios are based on
linear/convex programming and their (currently known) implementations do not
run in strongly polynomial time (this is true also for the algorithm of
\cite{Makarychev2014SOP}).
In contrast, the approximation framework developed in the present paper is
purely combinatorial with a strongly polynomial run-time.

\begin{table}
\center
\renewcommand{\arraystretch}{1.5}
\begin{tabular}{|c|c|c|c|c|c|}
\hline
paper & graphs & weights & algorithm & restrictions & approx.\ ratio \\
\hline
\multirow{2}{*}{\cite{Andrews2012RPM}} &
\multirow{2}{*}{undir.} &
\multirow{2}{*}{related} &
\multirow{2}{*}{math.\ prog.} &
\multirow{2}{*}{none} &
$O \left( 1 +
\left( \max_{e} \frac{\sigma_{e}}{\xi_{e}} \right)^{\frac{1}{\alpha}}
\log^{\alpha - 1} w_{\max}
\right)$ \\
& & & & & $O \left( N + \log^{\alpha - 1} w_{\max} \right)$ \\
\hline
\cite{Andrews2016MCN} &
undir. &
related &
math.\ prog. &
$\frac{\sigma_{e}}{\xi_{e}} = \frac{\sigma_{e'}}{\xi_{e'}}$ &
$\operatorname{polylog}(N) \cdot O \left( \log^{\alpha - 1} w_{\max} \right)$
\\
\hline
\cite{Antoniadis2014HHE} &
undir. &
$w_{i} = 1$ &
combin. &
$\xi_{e} = 1$ &
$O \left( \log^{\alpha} N \right)$ \\
\hline
\cite{Andrews2012RPM} &
undir. &
related &
math.\ prog. &
$\sigma_{e} = 0$ &
$O \left( \log^{\alpha - 1} w_{\max} \right)$ \\
\hline
\cite{Makarychev2014SOP} &
un/dir. &
unrel. &
math.\ prog. &
$\sigma_{e} = 0$ &
$(1 + \epsilon) \mathcal{B}_{\alpha}$ \\
\hline
current &
un/dir. &
unrel. &
combin. &
none &
$O \left(
1 +
\left( \max_{e} \frac{\sigma_{e}}{\xi_{e}} \right)^{\frac{1}{\alpha}}
\right)$ \\
\hline
\end{tabular}
\caption{\label{tab:comparison-routing-energy}%
Comparison of the approximation algorithms for GND with routing requests
(identifying the resources with the graph edges) under the energy consumption
cost functions
(\ref{equation:energy-consumption-cost-function}) and restrictions thereof.}
\end{table}

%%%%%%%%%%%%%%%%%%%%%%%%%%%%%%%%%%%%%%%
\paragraph{Scheduling Unrelated Parallel Machines.}
%%%%%%%%%%%%%%%%%%%%%%%%%%%%%%%%%%%%%%%
While GND with routing requests and related weights is a classic problem by its
own right, generalizing it to unrelated weights not also makes this
abstraction suitable for a wider class of GND scenarios, but also captures the
extensively studied problem of \emph{scheduling unrelated parallel machines}.
This problem can be represented as GND with routing requests over a graph
consisting of two vertices and multiple parallel edges (referred to as
\emph{machines}) between them.
The earlier algorithmic treatment of this problem considers the objective of
minimizing the $\ell_{\infty}$ norm (a.k.a.\ \emph{makespan}) of the machines'
load \cite{LenstraST1990, ShmoysT1993}.\footnote{%
This objective does not fit into the formulation of minimizing the sum of the
resource cost functions considered in our paper.
}
Later on, the focus has shifted to minimizing the $\ell_{p}$ norm of the
machines' load for
$p \in (1, \infty)$
\cite{AwerbuchAGKKV1995, AzarERW2004, AzarE2005, KumarMPS2009,
Makarychev2014SOP}.
The state of the art approximation algorithm in this regard is the one
developed by Kumar et al.~\cite{KumarMPS2009} with a $< 2$ approximation ratio
for all
$p \in (1, \infty)$.
Makarychev and Sviridenko \cite{Makarychev2014SOP} studied this problem for
small values of $p$ and designed a
$\sqrt[p]{\mathcal{B}_{p}}$-approximation,
improving upon the upper bound of \cite{KumarMPS2009} for the
$p \in (1, 2]$
regime.

The $\ell_{p}$ norm optimization criterion corresponds to the energy
consumption cost function (\ref{equation:energy-consumption-cost-function})
restricted to zero startup costs
$\sigma_{e} = 0$
(energy efficiency is also the main motivation of \cite{Makarychev2014SOP}).
In practice, however, machines' energy consumption typically incurs a positive
startup cost \cite{Andrews2016MCN, Antoniadis2014HHE}.
This motivated Khuller et al.~\cite{KhullerLS2010, LiK2011} to study a variant
of unrelated parallel machine scheduling in which the (sub)set of activated
machines should satisfy some budget constraint on the startup costs.
To the best of our knowledge, the current paper presents the first non-trivial
approximation algorithm for scheduling unrelated parallel machines that takes
into account the (positive) machines' startup costs
$\sigma_{e} > 0$
as part of the objective function.

%%%%%%%%%%%%%%%%%%%%%%%%%%%%%%%%%%%%%%%
\subsection{Paper Organization.}
%%%%%%%%%%%%%%%%%%%%%%%%%%%%%%%%%%%%%%%
The rest of the paper is organized as follows.
\Sect{}~\ref{sec:preliminaries} introduces the concepts and notations used
in the design and analysis of the proposed approximation framework.
Following that, a technical overview of the approximation framework's design
and analysis is provided in \Sect{}~\ref{sec:technical-overview}.
The actual approximation framework is presented in \Sect{}~\ref{sec:algorithm}
and analyzed in \Sect{}
\LongVersion %{
\ref{sec:analyzing-abrd-apx}--~\ref{sec:polyn-time-epsil}.
\LongVersionEnd %}
\ShortVersion %{
\ref{sec:analyzing-abrd-apx}--\ref{sec:potent-funct-shapl}.
\ShortVersionEnd %}
Two variants of the proposed approximation framework, which are more feasible
for a decentralized environment, are presented in
\LongVersion %{
\Sect{}~\ref{sec:impl-decentr-envir}.
\LongVersionEnd %}
\ShortVersion %{
the attached full version.
\ShortVersionEnd %}
\LongVersion %{
In \Sect{}~\ref{sec:poa-gnd-game}, we establish additional bounds that
demonstrate the tightness of certain components of the analysis.
Finally, alternative approaches for designing GND approximation algorithms are
discussed in \Sect{}~\ref{sec:altern-appr}. 
In particular, \Sect{}~\ref{subsection_convex_programming_and_rounding} discusses an alternative algorithm for the GND problem with routing requests using convex optimization and randomized rounding.
\LongVersionEnd %}
\ShortVersion %{
In that full version, we also establish additional bounds that demonstrate
the tightness of certain components of the analysis and discuss alternative
approaches for designing GND approximation algorithms.
\ShortVersionEnd %}

%%%%%%%%%%%%%%%%%%%%%%%%%%%%%%%%%%%%%%%%%%%%%%%%%%%%%%%%%%%%%%%%%%%%%%%%%%%%%%
\section{Preliminaries}
\label{sec:preliminaries}
%%%%%%%%%%%%%%%%%%%%%%%%%%%%%%%%%%%%%%%%%%%%%%%%%%%%%%%%%%%%%%%%%%%%%%%%%%%%%%
Throughout, we consider some GND problem $\mathcal{P}$ with succinctly
represented requests using REP resource cost functions
(\ref{equation:REP-cost-function}).
Let
\begin{MathMaybe}
\mathcal{I} =
\Left\langle
E,
\Left\{ P_{i}, \{ w_{i}(e) \}_{e \in E} \Right\}_{i \in [N]},
\Left\{ \alpha_{j} \Right\}_{j \in [q]},
\Left\{ \sigma_{e}, \Left\{ \xi_{e, j} \Right\}_{j \in [q]} \Right\}_{e \in E}
\Right\rangle
\end{MathMaybe}
be some $\mathcal{P}$ instance.
Let $p^{*}$ be an optimal solution for $\mathcal{I}$ and
$C^{*} = C(p^{*})$
be its total cost.

%%%%%%%%%%%%%%%%%%%%%%%%%%%%%%%%%%%%%%%
\paragraph{GND Games and Cost Sharing Mechanisms.}
%%%%%%%%%%%%%%%%%%%%%%%%%%%%%%%%%%%%%%%
A key ingredient of the approximation framework designed in this paper is
a \emph{GND game} derived from instance $\mathcal{I}$.
In this game, each request
$i \in [N]$
is associated with a strategic \emph{player $i$} that decides on the reply
$p_{i} \in P_{i}$
serving the request.
In the scope of this GND game, the reply
$p_{i} \in P_{i}$
is referred to as the \emph{strategy} of player $i$ and the reply collection
$P_{i}$ is referred to as her \emph{strategy space}.
We let
$P = (P_{1}, \dots, P_{N})$
and refer to
$p = (p_{1}, \dots, p_{N}) \in P$
as the (players') \emph{strategy profile}.
Although the strategy profile $p$ is a vector of replies, we may slightly
abuse the notation and write
$e \in p$
when we mean that
$e \in \bigcup_{i \in [N]} p_{i}$.

The cost $F_{e}(l_{e})$ of each resource
$e \in E$
is divided among the players based on a \emph{cost sharing mechanism (CSM)}
$M = \Left\{ f_{i, e}(\cdot) \Right\}_{i \in [N], e \in E}$,
where
$f_{i, e} : P \rightarrow \mathbb{R}_{\geq 0}$
is a \emph{cost sharing function} that determines the
\emph{cost share} $f_{i, e}(p)$ incurred by player
$i \in [N]$
for resource $e$ under strategy profile $p$, subject to the constraint that
$\sum_{i \in [N]} f_{i, e}(p) = F_{e}(l_{e}^{p})$.
Player $i$ chooses her strategy $p_{i}$ with the objective of minimizing her
\emph{individual cost}
$C_{i}(p) = \sum_{e \in E} f_{i, e}(p)$,
irrespective of the total cost
$C(p) = \sum_{i \in [N]} C_{i}(p)$
(a.k.a.\ the \emph{social cost}).

CSM
$M = \Left\{ f_{i, e}(\cdot) \Right\}_{i \in [N], e \in E}$
is said to be \emph{separable} and \emph{uniform} (cf.\
\cite{Chen2010DNP, Von2013OCS})
if the cost share of player
$i \in [N]$
in resource
$e \in E$
satisfies
(1)
if
$e \notin p_{i}$,
then
$f_{i, e}(p) = 0$;
and
(2)
$f_{i, e}(p)$ is fully determined by $w_{i}(e)$ and by the multiset of
weights of the (other) players using resource $e$.
Notice that $f_{i, e}(p)$ is independent of the identities and weights of the
players using any resource
$e' \neq e$.
It may be convenient to write
$f_{i, e}(S_{e})$
instead of
$f_{i, e}(p)$,
where
$S_{e} = \{ j \in [N] \mid e \in p_{j} \}$,
although, strictly speaking, $f_{i, e}(p)$ is also independent of the
identities (rather than weights) of the players in $S_{e} - \{ i \}$.
Unless stated otherwise, all CSMs considered in this paper
are \emph{separable} and \emph{uniform}.

%%%%%%%%%%%%%%%%%%%%%%%%%%%%%%%%%%%%%%%
\paragraph{Best Response.}
%%%%%%%%%%%%%%%%%%%%%%%%%%%%%%%%%%%%%%%
Following the convention in the game theoretic literature, given some
$i \in [N]$
and a strategy profile
$p = (p_{1}, \dots, p_{N})$,
let
$p_{-i} = (p_{1}, \dots, p_{i - 1}, p_{i + 1}, \dots, p_{N})$;
likewise, let
$P_{-i} = P_{1} \times \cdots \times P_{i - 1} \times
P_{i + 1} \times \cdots \times P_{N}$.
Given some approximation parameter
$\chi \geq 1$,
strategy
$p_{i} \in P_{i}$
is an \emph{approximate best response (ABR)} of player $i$ to
$p_{-i} \in P_{-i}$
if
$C_{i}(p_{i}, p_{-i}) \leq \chi \cdot C_{i}(p'_{i}, p_{-i})$
for every
$p'_{i} \in P_{i}$.
A \emph{best response (BR)} is an ABR with approximation parameter
$\chi = 1$.

A \emph{best response dynamic (BRD)} (resp., \emph{approximate best response
dynamic (ABRD)}) is an iterative procedure that given an initial strategy
profile
$p^{0} \in P$,
generates a sequence
$p^{1}, p^{2}, \dots$
of strategy profiles adhering to the rule that
for every
$t = 1, 2, \dots$,
there exists some
$i \in [N]$
such that
(1)
$p^{t}_{-i} = p^{t - 1}_{-i}$;
and
(2)
$p^{t}_{i}$ is a BR (resp., ABR) of player $i$ to $p^{t - 1}_{-i}$.

Strategy profile
$p \in P$
is a \emph{(pure) Nash equilibrium (NE)} of the GND game if $p_{i}$ is a BR to
$p_{-i}$ for every
$i \in [N]$.
The \emph{(pure) price of anarchy (PoA)} of the GND game is defined to be the
ratio 
$C(p) / C^{*}$,
where
$p \in P$
is a NE strategy profile that maximizes the social cost $C(p)$.

%%%%%%%%%%%%%%%%%%%%%%%%%%%%%%%%%%%%%%%
\paragraph{Smoothness.}
%%%%%%%%%%%%%%%%%%%%%%%%%%%%%%%%%%%%%%%
The following definition of Roughgarden \cite{Roughgarden2015IRP} plays a key
role in our analysis:
Given parameters
$\lambda > 0$
and
$0 < \mu < 1$,
we say that the GND game is \emph{$(\lambda, \mu)$-smooth} if
\begin{equation} \label{equation:smoothness-definition}
\sum_{i} C_{i}(p'_{i}, p_{-i})
\, \leq \,
\lambda C(p') + \mu C(p)
\end{equation}
for any two strategy profiles
$p, p' \in P$.\footnote{%
The original definition of Roughgarden~\cite{Roughgarden2015IRP}, that applies
for all cost minimization games, also requires that
$C(p) = \sum_{i \in [N]} C_{i}(p)$,
but this property is assumed to hold for all CSMs
considered in the current paper, so we do not mention it explicitly.
}
The game is said to be \emph{smooth} if it is $(\lambda, \mu)$-smooth for some
$\lambda > 0$
and
$0 < \mu < 1$.

%%%%%%%%%%%%%%%%%%%%%%%%%%%%%%%%%%%%%%%
\paragraph{Potential Functions.}
%%%%%%%%%%%%%%%%%%%%%%%%%%%%%%%%%%%%%%%
Function
$\Phi: P \rightarrow \mathbb{R}^{+}$
is said to be a \emph{potential function} if for every
$i \in [N]$
and for any two strategy profiles $p$ and $p'$ with
$p_{-i} = p'_{-i}$,
it holds that
\begin{MathMaybe}
\Phi(p') - \Phi(p)
\, = \,
C_{i}(p') - C_{i}(p) \, .
\end{MathMaybe}
A game admitting a potential function is said to be a \emph{potential game}.
The potential function $\Phi(p)$ is said to be \emph{$(A, B)$-bounded} for some
parameters
$A \geq 1$
and
$B \geq 1$
if
\begin{MathMaybe}
\Phi(p) / A
\, \leq \,
C(p)
\, \leq \,
B \cdot \Phi(p)
\end{MathMaybe}
for any strategy profile
$p \in P$.

%%%%%%%%%%%%%%%%%%%%%%%%%%%%%%%%%%%%%%%
\paragraph{Additional Notation and Terminology.}
%%%%%%%%%%%%%%%%%%%%%%%%%%%%%%%%%%%%%%%
Throughout, we think of $\epsilon > 0$ as a sufficiently small (positive)
constant and fix
$\epsilon_{1} = \frac{1 + \epsilon}{1 - \epsilon}$.
A probabilistic event $A$ is said to occur \emph{with high probability
(w.h.p.)} if
$\mathbb{P}(A) \geq 1 - 1 / (|E| + N)^{b}$,
where $b$ is an arbitrarily large constant.

%%%%%%%%%%%%%%%%%%%%%%%%%%%%%%%%%%%%%%%%%%%%%%%%%%%%%%%%%%%%%%%%%%%%%%%%%%%%%%
\section{Technical Overview}
\label{sec:technical-overview}
%%%%%%%%%%%%%%%%%%%%%%%%%%%%%%%%%%%%%%%%%%%%%%%%%%%%%%%%%%%%%%%%%%%%%%%%%%%%%%
The key concept in the design of our generic approximation framework is to
\emph{decouple} between the combinatorial structure of the specific GND
problem $\mathcal{P}$, captured by the request types (and encoded in the reply
collections), and the (D)oS cost functions of the individual resources.
Informally, the former is handled by the reply oracle
$\mathcal{O}_{\mathcal{P}}$ (specifically tailored for $\mathcal{P}$), whereas
for the latter, we harness the power of Roughgarden's \emph{smoothness} toolbox
\cite{Roughgarden2015IRP}.
Since this toolbox was originally introduced in the context of game theory
rather than algorithm design, we first transform the given $\mathcal{P}$
instance into a GND game by carefully choosing the CSM (more on that soon).
The algorithm then progresses via a sequence of player individual improvements
in the form of a BRD, where each BRD step is implemented by invoking
$\mathcal{O}_{\mathcal{P}}$ with a toll function constructed based on the
current strategy profile
$p \in P$,
the choice of player
$i \in [N]$,
and her cost sharing functions
$f_{i, e}(\cdot)$,
$e \in E$
(\Sect{}~\ref{sec:algorithm}).\footnote{%
In this section (only), we assume for simplicity that
$\mathcal{O}_{\mathcal{P}}$ is an exact reply oracle.
}

In order to establish the promised upper bound on the approximation ratio, we
first analyze the smoothness parameters of the aforementioned GND game
(\Sect{}~\ref{sec:smoothness})
which allows us to bound its PoA, thus ensuring that the total cost $C(p)$ of
any NE strategy profile
$p \in P$
provides the desired approximation for the (global) optimum $C^{*}$.
This part of the proof relies on introducing and analyzing a new class of
\emph{REP-expanded} CSMs
(\Sect{}~\ref{sec:smoothness}),
interesting in its own right.

One may hope that a BRD of the GND game converges to a NE strategy profile
$p \in P$,
but unfortunately, the BRD need not necessarily converge, and even if it does
converge, it need not necessarily be in polynomially many steps.
Inspired by another component of the smoothness toolbox
\cite{Roughgarden2015IRP} (which is in turn inspired by
\cite{Awerbuch2008FCN}), we show
(in \Sect{}~\ref{sec:analyzing-abrd-apx})
that if the game admits a bounded potential function, then after simulating
the BRD for polynomially many steps, one necessarily encounters a strategy
profile
$p \in P$
that yields the promised approximation
guarantee (although it is not necessarily a NE).

Does our GND game admit the desired bounded potential function?
The answer to this question depends, once again, on the choice of a CSM.
We therefore look for a CSM with three (possibly
conflicting) considerations in mind:
the game that it induces must admit a bounded potential function;
it must be REP-expanded;
and
it must be efficiently computable.
We prove that the \emph{Shapley CSM} satisfies the first
two conditions
(\Sect{} \ref{sec:potent-funct-shapl} and \ref{sec:smoothness}, respectively)
and although its exact computation is \#P-hard, we manage to adapt
\ShortVersion %{
(in the attached full version)
\ShortVersionEnd %}
the
approximation scheme of \cite{Liben2012CSV}, originally designed for
superadditive cost functions, to accommodate the REP cost functions
(\ref{equation:REP-cost-function}) with positive startup costs
$\sigma_{e} > 0$%
\LongVersion %{
\ (\Sect{}~\ref{sec:polyn-time-epsil})%
\LongVersionEnd %}
.
This presents another obstacle though since the original technique of
\cite{Roughgarden2015IRP} assumes (implicitly) that each step in the BRD is
(as the definition implies) an exact BR.
To overcome this obstacle, we show that an ABRD is still good enough for our
needs
(\Sect{}~\ref{sec:analyzing-abrd-apx}).

We believe that the construction described here demonstrates the usefulness of
algorithmic game theory tools for algorithm design even for optimization
problems that on the face of it, are not at all concerned with game theory.
A similar concept is demonstrated by Cole et al.~\cite{Cole2011IPS} who
obtained an improved combinatorial algorithm for job scheduling on unrelated
machines, with the objective of minimizing the weighted sum of completion
times, based on the game theoretic tools developed in \cite{Awerbuch2008FCN}.
In comparison, we employ the smoothness toolbox \cite{Roughgarden2015IRP} for
the design and analysis of our approximation framework.
It is the robustness of this toolbox that plays the key role in the generality
of our framework that can be applied to a wide family of GND problems.
This is in contrast to most of the existing approximation algorithms for such
problems that rely on linear/convex programming and are therefore heavily
tailored to one specific GND problem and much less generic.

%%%%%%%%%%%%%%%%%%%%%%%%%%%%%%%%%%%%%%%%%%%%%%%%%%%%%%%%%%%%%%%%%%%%%%%%%%%%%%
\section{Algorithm Description}
\label{sec:algorithm}
%%%%%%%%%%%%%%%%%%%%%%%%%%%%%%%%%%%%%%%%%%%%%%%%%%%%%%%%%%%%%%%%%%%%%%%%%%%%%%
Let $\mathcal{O}_{\mathcal{P}}$ be a reply $\varrho$-oracle for the requests
of the GND problem $\mathcal{P}$.
Our goal is to design an approximation algorithm with free access to
$\mathcal{O}_{\mathcal{P}}$ as promised in
\Thm{}~\ref{theorem:approximation-framework}.
We shall refer to this approximation algorithm as \texttt{Alg-ABRD}.

\LongVersion \sloppy \LongVersionEnd
Given an instance
$\mathcal{I} =
\Left\langle
E,
\Left\{ P_{i}, \{ w_{i}(e) \}_{e \in E} \Right\}_{i \in [N]},
\Left\{ \alpha_{j} \Right\}_{j \in [q]},
\Left\{ \sigma_{e}, \Left\{ \xi_{e, j} \Right\}_{j \in [q]} \Right\}_{e \in E}
\Right\rangle$
of $\mathcal{P}$, we first construct (conceptually) the GND game induced by
$\mathcal{I}$ and a carefully chosen CSM
$M = \Left\{ f_{i, e}(\cdot) \Right\}_{i \in [N], e \in E}$.
On top of the other properties of $M$ that will be discussed in the next
sections, we require that $M$ is \emph{poly-time $\epsilon$-computable},
namely, that given
$\mathcal{I}$,
$p \in P$,
and
$i \in [N]$,
it is possible to compute in time $\operatorname{poly}(|E|, N)$ some
\emph{$\epsilon$-cost shares} $\widetilde{f}_{i, e}(p)$,
$e \in E$,
that satisfy
\begin{MathMaybe}
(1 - \epsilon) f_{i, e}(p)
\leq
\widetilde{f}_{i, e}(p)
\leq
(1 + \epsilon) f_{i, e}(p)
\end{MathMaybe}
w.h.p.
Define the \emph{$\epsilon$-individual cost} $\widetilde{C}_{i}(p)$ to be
the sum
$\widetilde{C}_{i}(p) = \sum_{e \in E} \widetilde{f}_{i, e}(p)$,
which means that
\begin{MathMaybe}
(1 - \epsilon) C_{i}(p)
\leq
\widetilde{C}_{i}(p)
\leq
(1 + \epsilon)C_{i}(p)
\end{MathMaybe}
w.h.p.
As we shall perform the computations of the $\epsilon$-cost shares (and the
$\epsilon$-individual costs) $\operatorname{poly}(|E|, N)$ times, all of them
succeed w.h.p.;
condition hereafter on this event.
\LongVersion \par\fussy \LongVersionEnd

To simplify the presentation, we assume that the values of the $\epsilon$-cost
shares $\widetilde{f}_{i, e}(p)$,
$e \in E$,
and the $\epsilon$-individual costs $\widetilde{C}_{i}(p)$ have already been
fixed before the algorithm's execution for all
$i \in [N]$
and
$p \in P$
in an (arbitrary) manner that satisfies the aforementioned
$\epsilon$-approximation inequalities;
the algorithm then merely ``exposes'' some ($\operatorname{poly}(|E|, N)$
many) of these values.
The following lemma plays a key role in the design of \texttt{Alg-ABRD}.

\begin{lemma} \label{lemma:efficient-approximation-bar}
If $M$ is a poly-time $\epsilon$-computable CSM, then
there exists a randomized procedure that given
$i \in [N]$
and
$p_{-i} \in P_{-i}$,
runs in time $\operatorname{poly}(|E|, N)$ and computes a strategy
$p_{i} \in P_{i}$
and the corresponding $\epsilon$-individual cost
$\widetilde{C}_{i}(p_{i}, p_{-i})$
such that
$\widetilde{C}_{i}(p_{i}, p_{-i})
\leq
\varrho \cdot \widetilde{C}_{i}(p_{i}', p_{-i})$
for any
$p_{i}' \in P_{i}$.
This means in particular that $p_{i}$ is an ABR
of player $i$ to $p_{-i}$ with approximation parameter $\varrho
\epsilon_{1}$.\footnote{%
All subsequent occurrences of the term ABR (and
ABRD) share the same approximation parameter $\varrho \epsilon_{1}$, hence we
may refrain from mentioning this parameter explicitly.
}
\end{lemma}
\LongVersion %{
\begin{proof}
Construct the toll function
$\tau_{i, p_{-i}} : E \rightarrow \mathbb{R}_{\geq 0}$
by setting $\tau_{i, p_{-i}}(e)$ to be the $\epsilon$-cost share
$\widetilde{f}_{i, e}(S_{e} \cup \{ i \})$,
where
$S_{e} = \{ j \in [N] - \{ i \} \mid e \in p_{j} \}$.
This can be done in time $\operatorname{poly}(|E|, N)$ since $M$ is
poly-time $\epsilon$-computable.
The assumption that $M$ is separable and uniform guarantees that a reply
$p_{i} \in P_{i}$
that minimizes the total toll
$\sum_{e \in p_{i}} \tau_{i, p_{-i}}(e)$
up to factor $\varrho$ satisfies
$\widetilde{C}_{i}(p_{i}, p_{-i})
\leq
\varrho \cdot \widetilde{C}_{i}(p_{i}', p_{-i})$
for any
$p_{i}' \in P_{i}$
and that the sum
$\sum_{e \in p_{i}} \tau_{i, p_{-i}}(e)$
is the desired $\epsilon$-individual cost.
Such a reply $p_{i}$ can be computed using the reply $\varrho$-oracle
$\mathcal{O}_{\mathcal{P}}$.
\end{proof}
\LongVersionEnd %}

Employing the procedure promised by
\Lem{}~\ref{lemma:efficient-approximation-bar}, \texttt{Alg-ABRD} simulates an
ABRD
$p^{0}, p^{1}, \dots$
of the GND game induced by $\mathcal{I}$ and $M$ that includes at most $T$
iterations for some
$T = \operatorname{poly}(|E|, N)$
whose exact value will be determined later.
This is done as follows (see also
Pseudocode~\ref{algorithm_ABRD_for_EER_game}).

\begin{algorithm}
\caption{\texttt{Alg-ABRD}}\label{algorithm_ABRD_for_EER_game}
\begin{algorithmic}[1]
\State
\textbf{Input:} A GND instance
$\mathcal{I} =
\Left\langle
E,
\Left\{ P_{i}, \{ w_{i}(e) \}_{e \in E} \Right\}_{i \in [N]},
\Left\{ \alpha_{j} \Right\}_{j \in [q]},
\Left\{ \sigma_{e}, \Left\{ \xi_{e, j} \Right\}_{j \in [q]} \Right\}_{e \in E}
\Right\rangle$
\State
\textbf{Output:} A profile $p \in P$ that is feasible for the given instance $\mathcal{I}$
\ForAll{$i \in [N]$}
  \State{set $\tau_{i}^{0}$ to be the toll function defined by setting
$\tau_{i}^{0}(e) = F_{e}(w_{i}(e))$
for every
$e \in E$}
  \State{set $p_{i}^{0}$ to be the output of oracle
$\mathcal{O}_{\mathcal{P}}$ on $E$, $P_{i}$, and $\tau_{i}^{0}$}
\EndFor
\State{$t \leftarrow 0$}
\While{$t < T$}
  \State{$t \leftarrow t + 1$}
  \ForAll{$i \in [N]$}
    \State{set $p'_{i}$ to be an ABR of player $i$ to $p_{-i}^{t - 1}$ with
approximation parameter $\varrho \epsilon_{1}$}
    \State{$\delta_{i}^{t} \leftarrow
\widetilde{C}_{i}(p^{t - 1}) -
\epsilon_{1} \cdot \widetilde{C}_{i}(p'_{i}, p_{-i}^{t - 1})$}
  \EndFor
  \If{$\delta_{i}^{t} \leq 0$ for all $i \in [N]$}
    \State{$p^{t} \leftarrow p^{t - 1}$}
    \State{\texttt{break}}
  \EndIf
  \State{$\Delta^{t} \leftarrow \sum_{i \in [N]} \delta_{i}^{t}$}
  \State{pick some $j \in [N]$ such that
$\delta_{i}^{t} > 0$
and
$\delta_{i}^{t} \geq \Delta^{t} / N$}
  \State{$p^{t} \leftarrow (p'_{j}, p_{-j}^{t - 1})$}
\EndWhile
\State{$t^{*} = \operatorname{argmin}_{t} C(p^{t})$}
\State{return $p^{t^{*}}$}
\end{algorithmic}
\end{algorithm}

Set $p^{0}$ by taking $p^{0}_{i}$,
$i \in [N]$,
to be the strategy generated by $\mathcal{O}_{\mathcal{P}}$ for the toll
function $\tau_{i}^{0}$ defined by setting
$\tau_{i}^{0}(e) = F_{e}(w_{i}(e))$.
Assuming that $p^{t - 1}$,
$1 \leq t \leq T$,
was already constructed, we construct $p^{t}$ as follows.
For
$i \in [N]$,
employ the procedure promised by
\Lem{}~\ref{lemma:efficient-approximation-bar} to compute an ABR $p'_{i}$ of
player $i$ to
$p_{-i}^{t - 1}$
and let
$\delta_{i}^{t}
=
\widetilde{C}_{i}(p^{t - 1}) -
\epsilon_{1} \cdot \widetilde{C}_{i}(p'_{i}, p_{-i}^{t - 1})$.
If
$\delta_{i}^{t} \leq 0$
for all
$i \in [N]$,
then the ABRD stops, and we set
$p^{t} = p^{t - 1}$;
in this case, we say that the ABRD \emph{converges}.
Otherwise, fix
$\Delta^{t} = \sum_{i \in [N]} \delta_{i}^{t}$
and choose some player
$i' \in [N]$
so that
\begin{equation}
\label{formula_condition_of_being_selected_for_updating_the_strategy}
\delta_{i'}^{t} > 0
\quad \text{and} \quad
\delta_{i'}^{t} \geq \frac{1}{N}\Delta^{t}
\end{equation}
to update her strategy, setting
$p^{t} = (p'_{i'}, p_{-i'}^{t - 1})$
(the existence of such a player is guaranteed by the pigeonhole principle).

When the ABRD terminates (either because it has reached iteration
$t = T$
or because it converged), \texttt{Alg-ABRD} chooses an iteration $t^{*}$ such
that the corresponding strategy profile $p^{t^{*}}$ has a minimum total cost
$C(p^{t^{*}})$ and outputs $p^{t^{*}}$.
(Recall that in contrast to the player individual costs, the social cost can
always be computed efficiently.)

%%%%%%%%%%%%%%%%%%%%%%%%%%%%%%%%%%%%%%%%%%%%%%%%%%%%%%%%%%%%%%%%%%%%%%%%%%%%%%
\section{Analyzing \texttt{Alg-ABRD}}
\label{sec:analyzing-abrd-apx}
%%%%%%%%%%%%%%%%%%%%%%%%%%%%%%%%%%%%%%%%%%%%%%%%%%%%%%%%%%%%%%%%%%%%%%%%%%%%%%
In this section, we begin our journey towards bounding the approximation ratio
and run-time of \texttt{Alg-ABRD} as promised by
\Thm{}~\ref{theorem:approximation-framework}.
The analysis relies on a careful choice of the CSM
$M = \Left\{ f_{i, e}(\cdot) \Right\}_{i \in [N], e \in E}$.
In particular, we are looking for a CSM whose induced GND
game is smooth and admits a bounded potential function with the right choice
of parameters.
\LongVersion %{
The reason for that will be made clear in
\Thm{}~\ref{theorem:approx-ratio-Alg-ABRD} whose proof
relies on \Lem{}
\ref{lemma:ABRD-converge-approx-NE} and
\ref{lemma:optimal-reply-guarantees-approx-ratio};
the former provides an upper bound on the approximation ratio when the ABRD
converges, whereas the latter is used to bound the number $T$ of steps in the
ABRD (and is the key to ensuring strongly polynomial run-time).
\LongVersionEnd %}

\ShortVersion %{
\addtocounter{theorem}{1}
\ShortVersionEnd %}
\LongVersion %{
\begin{lemma}
\label{lemma:ABRD-converge-approx-NE}
Suppose that the CSM $M$ is chosen so that the induced GND
game is
$(\lambda, \mu)$-smooth with
$\mu < 1 / (\varrho \epsilon_{1}^{2})$.
If the ABRD simulated in \texttt{Alg-ABRD} converges at step $t$ for any
$t \in [T]$,
then the last strategy profile $p^{t}$ satisfies
\begin{MathMaybe}
C(p^{t})
\leq
\frac{\varrho \epsilon_{1}^{2}\lambda}{1 - \varrho \epsilon_{1}^2 \mu}
\cdot C^{*} \, .
\end{MathMaybe}
\end{lemma}
\LongVersionEnd %}
\LongVersion %{
\begin{proof}
Recalling that we use $p'_{i}$ to represent the ABR of player $i$ to
$p^{t}$, we observe that
\begin{align*}
C(p^{t})
\; = & \;
\sum_{i}C_{i}(p^{t}) \\
\leq & \;
\frac{1}{1 - \epsilon} \sum_{i}\widetilde{C}_{i}(p^{t}) \\
\leq & \;
\frac{1}{1 - \epsilon} \epsilon_{1} \sum_{i} \widetilde{C}_{i}(p'_{i},
p_{-i}^{t}) \\
\leq & \;
\frac{1}{1 - \epsilon} \epsilon_{1} \cdot \varrho \sum_{i} \widetilde{C}_{i}(p_{i}^{*},
p_{-i}^{t}) \\
\leq & \;
\varrho \epsilon_{1}^2 \cdot \sum_{i} C_{i}(p_{i}^{*}, p_{-i}^{t}) \\
\leq & \;
\varrho \epsilon_{1}^2 (\lambda \cdot C^{*} + \mu \cdot C(p^{t})) \, ,
\end{align*}
where
the second and fifth transitions follow from the definition of
$\epsilon$-individual cost,
the third transition holds since the algorithm converges at step $t$,
the fourth transition holds following
\Lem{}~\ref{lemma:efficient-approximation-bar},
and
the sixth transition follows from the definition of
$(\lambda, \mu)$-smoothness.
\end{proof}
\LongVersionEnd %}

\ShortVersion %{
\addtocounter{theorem}{1}
\ShortVersionEnd %}
\LongVersion %{
\begin{lemma}
\label{lemma:optimal-reply-guarantees-approx-ratio}
The initial strategy profile $p^{0}$ of \texttt{Alg-ABRD} satisfies
$C(p^{0}) \leq \varrho \cdot N^{\max_{j} \alpha_{j}} \cdot C^{*}$.
\end{lemma}
\LongVersionEnd %}
\LongVersion %{
\begin{proof}
The construction of $p^{0}$ guarantees that
\[
\sum_{e \in p_{i}^{0}} \Big[
\sigma_{e} + \sum_{j \in [q]} \xi_{e, j} (w_{i}(e))^{\alpha_{j}}
\Big]
\, \leq \,
\varrho \cdot \sum_{e \in p_{i}^{*}} \Big[
\sigma_{e} + \sum_{j \in [q]} \xi_{e, j} (w_{i}(e))^{\alpha_{j}}
\Big] \, .
\]
Therefore,
\begin{align*}
\sum_{i \in [N] } \sum_{e \in p_{i}^{0}} \Big[
\sigma_{e} + \sum_{j \in [q]}\xi_{e, j} (w_{i}(e))^{\alpha_{j}}
\Big]
\, \leq & \,
\varrho \cdot \sum_{i \in [N]} \sum_{e \in p_{i}^{*}} \Big[
\sigma_{e} + \sum_{j \in [q]} \xi_{e, j} (w_{i}(e))^{\alpha_{j}}
\Big] \\
\leq & \,
\varrho \cdot \sum_{e \in p^{*}} \Big[
N \cdot \sigma_{e} +
\sum_{j \in [q]} \xi_{e, j} \cdot \sum_{i : e \in p_{i}^{*}} (w_{i}(e))^{\alpha_{j}}
\Big] \\
\leq & \,
\varrho \cdot \sum_{e \in p^{*}} \bigg[
N \cdot \sigma_{e} +
\sum_{j \in [q]} \xi_{e, j} \cdot  \Big( \sum_{i : e \in p_{i}^{*}} w_{i}(e)
\Big)^{\alpha_{j}}
\bigg] \\
\leq & \,
\varrho \cdot N \sum_{e \in p^{*}} \bigg[
\sigma_{e} +
\sum_{j \in [q]} \xi_{e, j} \cdot  \Big( \sum_{i: e \in p_{i}^{*}} w_{i}(e)
\Big)^{\alpha_{j}}
\bigg] \\
= & \,
\varrho N \cdot C^{*} \, ,
\end{align*}
where the third transition follows from the superadditivity and the last
transition holds since
$l_{e}^{p*} > 0$
for every
$e \in p^{*}$.
Then,
\begin{align*}
C(p^{0})
\, =& \,
\sum_{e \in p^{0}} \bigg[
\sigma_{e} +
\sum_{j \in [q]} \xi_{e, j} \Big( \sum_{i: e \in p_{i}^{0}} w_{i}(e)
\Big)^{\alpha_{j}}
\bigg] \\
\leq & \,
\sum_{e \in p^{0}} \Big[
\sigma_{e} +
\sum_{j \in [q]} \xi_{e, j} \cdot N^{\alpha_{j} - 1} \sum_{i: e \in p_{i}^{0}}
(w_{i}(e))^{\alpha_{j}}
\Big] \\
\leq & \,
\sum_{e \in p^{0}} \Big[
\sum_{i: e \in p_{i}^{0}} \sigma_{e} +
\sum_{j \in [q]} \xi_{e, j} \cdot N^{\alpha_{j} - 1} \sum_{i: e \in p_{i}^{0}}
(w_{i}(e))^{\alpha_{j}}
\Big] \\
\leq & \,
N^{\max_{j}\alpha_{j} - 1} \sum_{e \in p^{0}} \sum_{i: e \in p_{i}^{0}} \Big[
\sigma_{e} + \sum_{j} \xi_{e, j} (w_{i}(e))^{\alpha_{j}} \Big] \\
= & \,
N^{\max_{j}\alpha_{j} - 1} \sum_{i \in [N] } \sum_{e \in p_{i}^{0}} \Big[
\sigma_{e} + \sum_{j} \xi_{e, j} (w_{i}(e))^{\alpha_{j}} \Big] \\
\leq & \,
\varrho N^{\max_{j}\alpha_{j}} \cdot C^{*} \, ,
\end{align*}
where the second transition holds because owing to the convexity,
\[
\bigg( \frac{\sum_{i: e \in p_{i}^{0}} w_{i}(e)}{|S_{e}^{p^{0}}|}
\bigg)^{\alpha_{j}}
\, \leq \,
\frac{1}{|S_{e}^{p^{0}}|} \sum_{i: e \in p_{i}^{0}} (w_{i}(e))^{\alpha_{j}} \, ,
\]
which means that
\[
\Big( \sum_{i: e \in p_{i}^{0}} w_{i}(e) \Big)^{\alpha_{j}}
\, \leq \,
|S_{e}^{p^{0}}|^{\alpha_{j} - 1} \sum_{i: e \in p_{i}^{0}}
(w_{i}(e))^{\alpha_{j}}
\, \leq \,
N^{\alpha_{j} - 1} \sum_{i: e \in p_{i}^{0}} (w_{i}(e))^{\alpha_{j}} \, .
\]
The assertion follows.
\end{proof}
\LongVersionEnd %}

\begin{theorem}
\label{theorem:approx-ratio-Alg-ABRD}
Suppose that the CSM $M$ is chosen so that the induced GND
game admits an $(A, B)$-bounded potential function $\Phi$ and is
$(\lambda, \mu)$-smooth
with
$\mu < 1 / (\varrho \epsilon_{1}^{2})$.
Let
$Q = \frac{2 \epsilon_{1} N A}{1 - \varrho \epsilon_{1}^{2} \mu}$.
If
$T =
\Left\lceil
Q \cdot \ln \Left( A B N^{\max_{j} \alpha_{j}} \Right)
\Right\rceil$,
then the output $p^{t^{*}}$ of \texttt{Alg-ABRD} satisfies
\begin{MathMaybe}
C(p^{t^{*}})
\leq
\frac{2 \varrho \epsilon_{1}^{2} \lambda}{1 - \varrho \epsilon_{1}^{2} \mu}
\cdot C^{*} \, .
\end{MathMaybe}
\end{theorem}
\LongVersion %{
\begin{proof}
\Lem{}~\ref{lemma:ABRD-converge-approx-NE} ensures that
the assertion holds if our ABRD converges at any step
$t \leq T$,
so it is left to analyze the case where the ABRD does not converge.
We say that profile $p^{t}$ of the ABRD is \emph{bad} if
\[
C(p^{t})
\, > \,
\frac{2 \varrho \epsilon_{1}^{2} \lambda}{1 - \varrho \epsilon_{1}^{2} \mu}
\cdot C^{*} \, .
\]

\begin{claim}
\label{proposition_potential_function_decreases_significantly_in_the_steps_with_bad_states}
For any
$t < T$,
if $p^{t}$ is bad, then
$\Phi(p^{t + 1}) < (1 - 1 / Q) \cdot \Phi(p^{t})$.
\end{claim}
\begin{subproof}
\qedlabel{proposition_potential_function_decreases_significantly_in_the_steps_with_bad_states}
Fix
\begin{equation}
\label{formula_difference_definition_of_nabla_tau}
d^{t}
\; = \;
\frac{1}{1 - \epsilon} \Big[
\sum_{i \in [N]} \widetilde{C}_{i}(p^{t}) - \varrho \cdot
\epsilon_{1}  \sum_{i \in [N]} \widetilde{C}_{i}(p_{i}^{*}, p_{-i}^{t})
\Big] \, .
\end{equation}
This means that
\begin{align*}
C(p^{t})
\; = \;
\sum_{i \in [N]} C_{i}(p^{t})
\; \leq & \;
\frac{1}{1 - \epsilon} \sum_{i \in [N]} \widetilde{C}_{i}(p^{t}) \\
= & \;
\frac{\varrho \epsilon_{1}}{1 - \epsilon} \sum_{i \in [N]}
\widetilde{C}_{i}(p_{i}^{*}, p_{-i}^{t}) + d^{t} \\
\leq & \;
\varrho \epsilon_{1}^{2} \sum_{i \in [N]} C_{i}(p_{i}^{*}, p_{-i}^{t}) + d^{t} \\
\leq & \;
\varrho \epsilon_{1}^{2} (\lambda \cdot C^{*} + \mu C(p^{t})) +
d^{t} \, .
\end{align*}
Therefore,
$d^{t}
\geq
\Left[ 1 - \varrho \epsilon_{1}^{2} \mu \Right] C(p^{t}) -
\varrho \epsilon_{1}^{2} \lambda \cdot C^{*}$,
hence, if $p^{t}$ is bad, then $d^{t}$ satisfies
\begin{equation}
\label{formula_lower_bound_on_nabla_tau}
d^{t}
\, > \,
\Left[ 1 - \varrho \epsilon_{1}^{2} \mu \Right] C(p^{t}) -
\frac{1 - \varrho \epsilon_{1}^{2} \mu}{2} C(p^{t})
\, = \,
\frac{1 - \varrho \epsilon_{1}^{2} \mu}{2} C(p^{t}) \, .
\end{equation}

Since the ABRD does not converge at step $t$, there exists a player $i^{t}$
being selected to update its strategy.
Recalling that the ABR of player $i$ to $p^{t}$ is
denoted by $p'_{i}$, we observe that
\begingroup
\allowdisplaybreaks 
\begin{align*}
\Phi(p^{t}) - \Phi(p^{t + 1})
\; = & \;
C_{i^{t}}(p^{t}) - C_{i^{t}}(p'_{i^{t}}, p_{-i^{t}}^{t}) \\
\geq & \;
\frac{1}{1 + \epsilon} \widetilde{C}_{i^{t}}(p^{t}) - \frac{1}{1 - \epsilon}
\widetilde{C}_{i^{t}}(p'_{i^{t}}, p_{-i^{t}}^{t}) \\
= & \;
\frac{1}{1 + \epsilon} \Left[ \widetilde{C}_{i^{t}}(p^{t}) - \epsilon_{1}
\widetilde{C}_{i^{t}}(p'_{i^{t}}, p_{-i^{t}}^{t}) \Right] \\
\geq & \;
\frac{1}{1 + \epsilon} \cdot \frac{1}{N} \sum_{i \in [N]} \Left[
\widetilde{C}_{i}(p^{t}) - \epsilon_{1} \widetilde{C}_{i}(p'_{i},
p_{-j}^{t}) \Right] \\
\geq & \;
\frac{1}{1 + \epsilon} \cdot \frac{1}{N} \sum_{i \in [N]} \Left[
\widetilde{C}_{i}(p^{t}) - \varrho \cdot \epsilon_{1} \widetilde{C}_{i}(p_{i}^{*},
p_{-j}^{t}) \Right] \\
= & \;
\frac{1 - \epsilon}{1 + \epsilon} \cdot \frac{d^{t}}{N} \\
> & \;
\frac{1 - \epsilon}{1 + \epsilon} \cdot \frac{1}{2N} \Left[ 1 -
\varrho \epsilon_{1}^{2} \mu \Right] C(p^{t}) \\
\geq & \;
\frac{1}{\epsilon_{1}} \cdot \frac{1}{2N} \Left[1 - \varrho \epsilon_{1}^{2} \mu
\Right] \frac{\Phi(p^{t})}{A} \, ,
\end{align*}
\endgroup
where
the fourth transition follows from
Eq.~(\ref{formula_condition_of_being_selected_for_updating_the_strategy}),
the fifth transition holds since $p'_{i^{t}}$ is the ABR promised by Lemma \ref{lemma:efficient-approximation-bar}, which means that
$\widetilde{C}_{i^{t}}(p'_{i^{t}}, p_{-i^{t}}^{t}) \leq \varrho \cdot \widetilde{C}_{i^{t}}(p_{i^{t}}^{*},
p_{-i^{t}}^{t})$,
the sixth and seventh transitions follow from
Eq.~(\ref{formula_difference_definition_of_nabla_tau}) and
Eq.~(\ref{formula_lower_bound_on_nabla_tau}), respectively,
and
the last transition holds because the potential function is assumed to be $(A,
B)$-bounded.
Therefore,
\[
\Phi(p^{t + 1})
\, < \,
\Phi(p^{t}) \Left( 1 - \frac{1 - \varrho \epsilon_{1}^{2} \mu}{2 \epsilon_{1} N
A} \Right)
\, = \,
(1 - 1 / Q) \cdot \Phi(p^{t})
\]
as promised.
\end{subproof}

Since \texttt{Alg-ABRD} outputs the strategy profile with the minimum total
cost among all the generated strategy profiles, this theorem holds if any of
these strategy profiles is not bad.

\begin{claim}
\label{proposition_decrease_in_potential_function_when_all_states_are_bad}
If all the
$T + 1$
strategy profiles in the ABRD are bad, then
$C(p^{T}) < \varrho \cdot C^{*}$.
\end{claim}
\begin{subproof}
\qedlabel{proposition_decrease_in_potential_function_when_all_states_are_bad}
Claim~\ref{proposition_potential_function_decreases_significantly_in_the_steps_with_bad_states}
implies that if all the $T + 1$ profiles in the ABRD are bad, then
\[
\Phi(p^{T})
\, < \,
\Left( 1 - \frac{1}{Q} \Right)^{T} \Phi(p^{0})
\, = \,
\Left( 1 - \frac{1}{Q} \Right)^{\Left\lceil Q \cdot \ln \Big( A B
N^{\max_{j}\alpha_{j}} \Big) \Right\rceil} \Phi(p^{0})
\, \leq \,
\frac{1}{A B N^{\max_{j}\alpha_{j}}} \Phi(p^{0}) \, .
\]
By the definition of the bounded potential function and by
\Lem{}~\ref{lemma:optimal-reply-guarantees-approx-ratio},
we have
\begin{equation*}
C(p^{T})
\, \leq \,
B \cdot \Phi(p^{T})
\, < \,
\frac{B}{A B N^{\max_{j}\alpha_{j}}} \Phi(p^{0})
\, \leq \,
\frac{A}{A N^{\max_{j}\alpha_{j}}} C(p^{0})
\, \leq \,
\frac{\varrho N^{\max_{j}\alpha_{j}} C^{*}}{N^{\max_{j}\alpha_{j}}}
\, = \,
\varrho \cdot C^{*}
\end{equation*}
which completes the proof.
\end{subproof}

\begin{claim}
\label{proposition_general_lower_bound_on_smoothness_parameters}
$\varrho
<
\frac{2\varrho \epsilon_{1}^{2} \lambda}
{1 - \varrho \epsilon_{1}^{2}\mu}$.
\end{claim}
\begin{subproof}
\qedlabel{proposition_general_lower_bound_on_smoothness_parameters}
For any cost minimization $(\lambda, \mu)$-smooth game that has a (bounded)
potential function, we have
$\frac{\lambda}{1 - \mu} \geq 1$.
This is because the existence of a potential function implies the existence
of a (pure) NE
$p \in P$
with
$C(p) \leq \frac{\lambda}{1 - \mu} C^{*}$
\cite{Roughgarden2015IRP}.
Therefore,
$\frac{2 \varrho \epsilon_{1}^{2} \lambda}{1 - \varrho \epsilon_{1}^{2} \mu}
>
\frac{2\varrho \lambda}{1 - \mu}
>
\varrho$.
\end{subproof}

By combining Claims
\ref{proposition_decrease_in_potential_function_when_all_states_are_bad} and
\ref{proposition_general_lower_bound_on_smoothness_parameters}, we conclude
that not all $T + 1$ profiles are bad, thus completing the proof.
\end{proof}
\LongVersionEnd %}

\ShortVersion %{
\addtocounter{theorem}{1}
\ShortVersionEnd %}
\LongVersion %{
\begin{remark}
Roughgarden~\cite{Roughgarden2015IRP} proves that in the BRD of a
$(\lambda, \mu)$-smooth
game, the number of strategy profiles whose cost is larger than
$\frac{\lambda}{(1 - \upsilon)(1 - \mu)} \cdot C^{*}$
for some constant
$\upsilon \in (0, 1)$
is bounded by a polynomial.
However, his proof depends on the exact values of the cost shares and exact
best responses, both of which may be intractable in our GND setting.
\end{remark}
\LongVersionEnd %}

In the following sections, we search for a CSM whose induced GND game is
$(\lambda, \mu)$-smooth and admits an $(A, B)$-bounded potential function for
parameters $\lambda$, $\mu$, $A$, and $B$ that when plugged into
\Thm{}~\ref{theorem:approx-ratio-Alg-ABRD}, yield the
desired approximation ratio and run-time bounds.

%%%%%%%%%%%%%%%%%%%%%%%%%%%%%%%%%%%%%%%%%%%%%%%%%%%%%%%%%%%%%%%%%%%%%%%%%%%%%%
\section{Smoothness of the GND Game}
\label{sec:smoothness}
%%%%%%%%%%%%%%%%%%%%%%%%%%%%%%%%%%%%%%%%%%%%%%%%%%%%%%%%%%%%%%%%%%%%%%%%%%%%%%
In this section, a rather wide class of CSMs is presented
and the smoothness parameters of the induced GND games are analyzed.
The proof that an adequate potential function exists for (the GND game induced
by) one of these CSMs is deferred to
\Sect{}~\ref{sec:potent-funct-shapl}.

A CSM (for GND games) is said to be \emph{REP-expanded} if the cost share
$f_{i, e}(p)$ satisfies
\[
f_{i, e}(p)
\, \leq \,
\sigma_{e} +
\sum_{j \in [q]} \xi_{e, j} \sum_{k = 1}^{K_{j}} z_{k, j}
\Left( l_{e}^{p} - w_{i}(e) \Right)^{x_{k, j}} (w_{i}(e))^{y_{k, j}} \, ,
\]
for any
player
$i \in [N]$,
edge
$e \in E$,
and strategy profile
$p \in P$,
where
$K_{j}$,
$0 \leq x_{k, j} \leq \alpha_{j} - 1$,
$1 \leq y_{k, j} \leq \alpha_{j}$,
and
$z_{k, j}$,
$k \in [K_{j}]$,
are non-negative constants that can only depend on $\alpha_{j}$;
moreover, we require that
$x_{k, j} + y_{k, j} = \alpha_{j}$.
Note that the exponents $x_{k, j}$ and $y_{k, j}$ and the coefficient
$z_{k, j}$ are not necessarily integral.

\LongVersion \sloppy \LongVersionEnd
\begin{theorem}
\label{theorem_smoothness_of_polynomial_bounded_cost_sharing_mechanisms}
Consider some REP-expanded CSM $M$.
For any
$\varrho \geq 1$,
the GND game induced by $M$ is guaranteed to be
$\Left(
\gamma_{\alpha} + \lambda_{\alpha} \cdot \varrho^{\max_{j}\alpha_{j}},
1 / (2 \varrho)
\Right)$-smooth, 
where
$\gamma_{\alpha} =
\max_{e \in E} \min_{j \in [q]} \Left( \frac{1}{\alpha_{j} - 1} \cdot
\frac{\sigma_{e}}{\xi_{e, j}} \Right)^{1 / \alpha_{j}}$
and
$\lambda_{\alpha} > 0$
is a positive constant that depends only on $q$ and
$\alpha_{1}, \dots, \alpha_{q}$.
\end{theorem}
\LongVersion \par\fussy \LongVersionEnd
\begin{proof}
\LongVersion %{
Our goal in this proof is to show that
(\ref{equation:smoothness-definition}) holds with
$\lambda = \gamma_{\alpha} + \lambda_{\alpha} \cdot \varrho^{\max_{j}\alpha_{j}}$
and
$\mu = 1 / (2 \varrho)$.
\LongVersionEnd %}
We begin by observing that
\begin{align*}
\sum_{i \in [N]} C_{i}(p'_{i}, p_{-i})
\, = & \,
\sum_{i \in [N]} \sum_{e \in p'_{i}} f_{i, e}(p'_{i}, p_{-i}) \\
\leq & \,
\sum_{i \in [N]} \sum_{e \in p'_{i}} \Left( \sigma_{e} + \sum_{j \in [q]}
\xi_{e, j} \sum_{k = 1}^{K_{j}} z_{k, j} (l_{e}^{p})^{x_{k, j}}
(w_{i}(e))^{y_{k, j}} \Right ) \\
\leq & \,
\sum_{i \in [N]} \sum_{e \in p'_{i}} \sigma_{e}
+
\sum_{j \in [q], k \in [K_{j}]} z_{k, j} \sum_{e \in p'} \xi_{e, j}
(l_{e}^{p})^{x_{k, j}} \sum_{i : e \in p'_{i}} (w_{i}(e))^{y_{k, j}} \\
\leq & \,
\sum_{e \in p'} \sigma_{e} l_{e}^{p'}
+
\sum_{j \in [q], k \in [K_{j}]} z_{k, j} \sum_{e \in p'} \xi_{e, j} (l_{e}^{p})^{x_{k, j}}
(l_{e}^{p'})^{y_{k, j}} \, ,
\end{align*}
where
the second transition follows by the definition of REP-expanded CSMs because
when player $i$ deviates to $p'_{i}$, the load on edge
$e \in p'_{i}$
is at most
$l_{e}^{p} + w_{i}(e)$
and
the last transition holds because
(1)
$w_{i}(e) \geq 1$,
hence
$|\{ i \in [N] : e \in p'_{i}\}| \leq l_{e}^{p'}$
for any edge $e$;
and
(2)
$(w_{i}(e))^{y_{k,\, j}}$ is a superadditive function of $w_{i}(e)$, hence
$\sum_{i : e \in p'_{i}} (w_{i}(e))^{y_{k,\, j}}
\leq
(\sum_{i : e \in p'_{i}} w_{i}(e))^{y_{k,\, j}}
=
(l_{e}^{p'})^{y_{k}}$.
The desired upper bound on
$\sum_{e \in p'} \sigma_{e} l_{e}^{p'} + \sum_{j, k} z_{k,\, j} \sum_{e \in p'}
\xi_{e,\, j} (l_{e}^{p})^{x_{k,\, j}} (l_{e}^{p'})^{y_{k,\, j}}$
is established in Claims
\ref{claim:smooth-analysis-upper-bound-static-part} and
\ref{claim:smooth-analysis-upper-bound-dynamic-part}.

\begin{claim}
\label{claim:smooth-analysis-upper-bound-static-part}
$\displaystyle \sum_{e \in p'}\sigma_{e}l_{e}^{p'} \leq \gamma_{\alpha} \cdot C(p')$
\end{claim}
\LongVersion %{
\begin{subproof}
\qedlabel{claim:smooth-analysis-upper-bound-static-part}
Define the function
$g(x) = \frac{\sigma x}{\sigma + \xi x^{\alpha}}$ for arbitrary positive numbers $\sigma > 0$, $\xi > 0$ and $\alpha > 1$.
Since its derivative is
\begin{MathMaybe}
g'(x)
\; = \;
\frac{\sigma}{(\sigma + \xi x^{\alpha})^2}
[\sigma - (\alpha - 1)\xi x^{\alpha}] \, ,
\end{MathMaybe}
it attains its maximum for
$x \geq 0$
at
$x = \Left( \frac{\sigma}{\xi (\alpha - 1)} \Right)^{1 / \alpha}$.
Therefore, for any
$x \geq 0$,
we have
\[
\frac{\sigma \cdot x}{\sigma + \xi \cdot x^{\alpha}}
\, = \,
g(x)
\, \leq \,
g \Left( (\sigma / \xi (\alpha - 1))^{1 / \alpha} \Right)
\, = \,
\frac{\sigma \Left( \frac{\sigma}{\xi (\alpha - 1)} \Right)^{1 /
\alpha}}{\sigma + \xi \frac{\sigma}{\xi (\alpha - 1)}}
\, = \,
\Left( \frac{1}{\alpha - 1} \cdot \frac{\sigma}{\xi} \Right)^{1 / \alpha}
\bigg/
\Left( 1 + \frac{1}{\alpha - 1} \Right) \, .
\]
Let
$j_{e}^{*} \in \operatorname{argmin}_{j \in [q]}\Left( \frac{1}{\alpha_{j} -
1}\cdot\frac{\sigma_{e}}{\xi_{e,\, j}} \Right)^{1/\alpha_{j}}$.
By the inequality above, we have
\begin{MathMaybe}
\sigma_{e} l_{e}^{p'}
\, < \,
\Left(
\frac{1}{\alpha_{j_{e}^{*}} - 1} \cdot \frac{\sigma_{e}}{\xi_{e, j_{e}^{*}}}
\Right)^{1 / \alpha_{j_{e}^{*}}} \cdot
\Left[
\sigma_{e} + \xi_{e, j_{e}^{*}} (l_{e}^{p'})^{1 / \alpha_{j_{e}^{*}}}
\Right]
\, \leq \,
\Left(
\frac{1}{\alpha_{j_{e}^{*}} - 1} \cdot \frac{\sigma_{e}}{\xi_{e, j_{e}^{*}}}
\Right)^{1 / \alpha_{j_{e}^{*}}} \cdot F_{e}(l_{e}^{p'}) \, .
\end{MathMaybe}
The assertion follows since
$\Left(
\frac{1}{\alpha_{j_{e}^{*}} - 1} \cdot \frac{\sigma_{e}}{\xi_{e, j_{e}^{*}}}
\Right)^{1/\alpha_{j_{e}^{*}}}
\leq
\gamma_{\alpha}$
for every
$e \in E$.
\end{subproof}
\LongVersionEnd %}

Fix
$z_{\max} = \lceil \max_{j,\, k} z_{k,\, j} \rceil$
and let
$\lambda_{\alpha} = \max_{j \in [q]}(2K_{j} \cdot z_{\max})^{\max_{j}\alpha_{j} + 1}$.

\begin{claim}
\label{claim:smooth-analysis-upper-bound-dynamic-part}
$\sum_{j = 1}^{q} \sum_{k = 1}^{K_{j}} z_{k, j} \sum_{e \in p'} \xi_{e, j}
(l_{e}^{p})^{x_{k, j}} (l_{e}^{p'})^{y_{k, j}}
\leq
\lambda_{\alpha} \varrho^{\max_{j} \alpha_{j}} C(p') + C(p) / (2 \varrho)$.
\end{claim}
\begin{subproof}
\qedlabel{claim:smooth-analysis-upper-bound-dynamic-part}
Let $p$ and $p'$ be any two profiles.
First, consider the terms
$\sum_{e \in p'} \xi_{e,\, j} (l_{e}^{p})^{x_{k,\, j}} (l_{e}^{p'})^{y_{k,\, j}}$
with
$x_{k,\, j} > 0$.
These terms satisfy
\begin{align*}
\sum_{e \in p'} \xi_{e, j} (l_{e}^{p})^{x_{k, j}} (l_{e}^{p'})^{y_{k, j}}
\, = & \,
\sum_{e \in p \cap p'} \xi_{e, j} (l_{e}^{p})^{x_{k, j}} (l_{e}^{p'})^{y_{k, j}} \\
= & \,
\sum_{e \in p \cap p'} \Left[ (\xi_{e, j})^{\frac{x_{k, j}}{\alpha}}
(l_{e}^{p})^{x_{k, j}}
\Right] \cdot
\Left[ (\xi_{e, j})^{\frac{y_{k, j}}{\alpha}} (l_{e}^{p'})^{y_{k, j}} \Right] \\
\leq & \,
\Left\lVert (\vec{\xi_{j}} \,)^{\frac{x_{k, j}}{\alpha}} (\vec{l}^{\:p})^{x_{k, j}}
\Right\rVert_{\frac{\alpha_{j}}{x_{k, j}}} \cdot \Left\lVert
(\vec{\xi_{j}}\,)^{\frac{y_{k, j}}{\alpha}} (\vec{l}^{\:p'})^{y_{k, j}}
\Right\rVert_{\frac{\alpha_{j}}{y_{k_{j}}}} \\ 
= & \,
\Left[ \sum_{e \in p \cap p'} \xi_{e, j} (l_{e}^{p})^{\alpha_{j}}
\Right]^{\frac{x_{k, j}}{\alpha_{j}}} \cdot
\Left[ \sum_{e \in p \cap p'} \xi_{e, j} (l_{e}^{p'})^{\alpha_{j}}
\Right]^{\frac{y_{k, j}}{\alpha_{j}}} \, ,
\end{align*}
where $\vec{\xi_{j}}$, $\vec{l}^{\: p}$, $\vec{l}^{\: p'}$ represent the vectors
composed from $\xi_{e,\, j}$, $l_{e}^{p}$, and $l_{e}^{p'}$, respectively, for
$e \in p \cap p'$,
the first transition holds since
$(l_{e}^{p})^{x_{k,\, j}} = 0$
whenever
$e \in p' - p$,
and
the third transition follows from H\"{o}lder's inequality.
Now consider the terms
$\sum_{e \in p'} \xi_{e,\, j} (l_{e}^{p})^{x_{k,\, j}} (l_{e}^{p'})^{y_{k,\, j}}$
with
$x_{k,\, j} = 0$
which means that
$y_{k,\, j} = \alpha_{j}$
since
$x_{k,\, j} + y_{k,\, j} = \alpha_{j}$.
For these terms, we also have
\begin{MathMaybe}
\sum_{e \in p'} \xi_{e, j} (l_{e}^{p})^{x_{k, j}} (l_{e}^{p'})^{y_{k, j}}
\, = \,
\Left[
\sum_{e \in p \cap p'} \xi_{e, j} (l_{e}^{p})^{\alpha_{j}}
\Right]^{\frac{x_{k, j}}{\alpha_{j}}} \cdot
\Left[
\sum_{e \in p \cap p'} \xi_{e, j} (l_{e}^{p'})^{\alpha_{j}}
\Right]^{\frac{y_{k, j}}{\alpha_{j}}} \, .
\end{MathMaybe}

\ShortVersion \sloppy \ShortVersionEnd
So,
\begin{MathMaybe}
\sum_{k = 1}^{K_{j}} z_{k, j}
\sum_{e \in p'}
\xi_{e, j} (l_{e}^{p})^{x_{k, j}} (l_{e}^{p'})^{y_{k, j}}
\, \leq \,
\sum_{k = 1}^{K_{j}}
z_{k, j}
\Left(
\vartheta_{j}(p)
\Right)^{\frac{x_{k, j}}{\alpha_{j}}} \cdot
\Left(
\vartheta_{j}(p')
\Right)^{\frac{y_{k, j}}{\alpha_{j}}} \, ,
\end{MathMaybe}
where
$\vartheta_{j}(p) =
\sum_{e \in p \cap p'} \xi_{e, j} (l_{e}^{p})^{\alpha_{j}}$
and
$\vartheta_{j}(p') =
\sum_{e \in p \cap p'} \xi_{e, j} (l_{e}^{p'})^{\alpha_{j}}$.
We proceed to analyze the upper bound on
\begin{MathMaybe}
\sum_{k = 1}^{K_{j}} z_{k, j}
\Left(
\vartheta_{j}(p)
\Right)^{\frac{x_{k, j}}{\alpha_{j}}} \cdot
\Left(
\vartheta_{j}(p')
\Right)^{\frac{y_{k, j}}{\alpha_{j}}}
\end{MathMaybe}
by considering the following two cases.
\begin{DenseItemize}

\item
$\vartheta_{j}(p)
<
(2 K_{j} z_{\max} \varrho)^{\alpha_{j}} \cdot \vartheta_{j}(p')$:
In this case, we have
\begin{MathMaybe}
\sum_{k = 1}^{K_{j}}
z_{k, j}
\Left( \vartheta_{j}(p) \Right)^{\frac{x_{k, j}}{\alpha_{j}}}
\Left( \vartheta_{j}(p') \Right)^{\frac{y_{k, j}}{\alpha_{j}}}
\, \leq \,
\sum_{k = 1}^{K_{j}}
z_{k, j} (2 K_{j} z_{\max} \varrho)^{x_{k, j}} \vartheta_{j}(p')
\, < \,
K_{j} z_{\max} (2 K_{j} z_{\max} \varrho)^{\alpha_{j}} \cdot \vartheta_{j}(p')
\, ,
\end{MathMaybe}
where
the second transition holds because
$x_{k, j} < \alpha_{j}$
and
$2 K_{j} z_{\max} \varrho > 1$.

\item
$\vartheta_{j}(p)
\geq
(2 K_{j} z_{\max} \varrho)^{\alpha_{j}} \cdot \vartheta_{j}(p')$:
In this case, we have
\begin{MathMaybe}
\sum_{k = 1}^{K_{j}}
z_{k,\, j}
(\vartheta_{j}(p))^{\frac{x_{k,\, j}}{\alpha_{j}}}
(\vartheta_{j}(p'))^{\frac{y_{k,\, j}}{\alpha_{j}}}
\, \leq \,
\sum_{k = 1}^{K_{j}}
z_{k,\, j} \vartheta_{j}(p) \cdot
\frac{1}{(2 K_{j} z_{\max} \varrho)^{y_{k,\, j}}}
\, \leq \,
\frac{1}{2\varrho} \vartheta_{j}(p) \, ,
\end{MathMaybe}
where
the last transition holds because
$y_{k,\, j} \geq 1$
and
$2 K_{j} z_{\max} \varrho> 1$.

\end{DenseItemize}
Therefore,
\begin{align*}
\sum_{j, k} z_{k, j}
(\vartheta_{j}(p))^{\frac{x_{k, j}}{\alpha_{j}}}
(\vartheta_{j}(p'))^{\frac{y_{k, j}}{\alpha_{j}}}
\, \leq \, &
\sum_{j = 1}^{q}
\lambda_{\alpha} \varrho^{\max_{j}\alpha_{j}} \vartheta_{j}(p') +
\vartheta_{j}(p) / (2 \varrho) \\
= \, &
\lambda_{\alpha} \varrho^{\max_{j}\alpha_{j}}
\Left(
\sum_{j = 1}^{q} \sum_{e \in p \cap p'}
\xi_{e, j} (l_{e}^{p'})^{\alpha_{j}}
\Right)
+
\frac{1}{2 \varrho}
\Left(
\sum_{j = 1}^{q} \sum_{e \in p \cap p'}
\xi_{e, j} (l_{e}^{p'})^{\alpha_{j}}
\Right) \\
< \, &
\lambda_{\alpha} \varrho^{\max_{j} \alpha_{j}}
\Left[
\sum_{e \in p'}
\Left(
\sigma_{e} + \sum_{j = 1}^{q} \xi_{e, j} (l_{e}^{p'})^{\alpha_{j}}
\Right)
\Right] \\
& +
\frac{1}{2 \varrho}
\Left[
\sum_{e \in p}
\Left(
\sigma_{e} + \sum_{j = 1}^{q} \xi_{e, j} (l_{e}^{p})^{\alpha_{j}}
\Right)
\Right] \\
= \, &
\lambda_{\alpha} \varrho^{\max_{j} \alpha_{j}} C(p') + C(p) / (2 \varrho)
\end{align*}
which establishes the assertion.
\end{subproof}
\ShortVersion \par\fussy \ShortVersionEnd

Together, Claims \ref{claim:smooth-analysis-upper-bound-static-part}
and \ref{claim:smooth-analysis-upper-bound-dynamic-part} imply that
\LongVersion %{
\[
\sum_{i \in [N]} C_{i}(p'_{i}, p_{-i})
\, \leq \,
\Left( \gamma_{\alpha} + \lambda_{\alpha}\varrho^{\max_{j}\alpha_{j}} \Right) C(p') +
C(p) / (2 \varrho) \, ,
\]
so
\LongVersionEnd %}
(\ref{equation:smoothness-definition}) indeed holds with
$\lambda = \gamma_{\alpha} + \lambda_{\alpha}\varrho^{\max_{j}\alpha_{j}}$
and
$\mu = 1 / (2 \varrho)$.
\end{proof}

Since
$1 / (2 \varrho) < 1 / (\varrho \epsilon_{1}^{2})$
for sufficiently small
$\epsilon > 0$,
it follows that we can employ
\Thm{}~\ref{theorem:approx-ratio-Alg-ABRD} with the
smoothness parameters
\LongVersion %{
$\lambda = \gamma_{\alpha} + \lambda_{\alpha} \cdot \varrho^{\max_{j}\alpha_{j}}$
and
$\mu = 1 / (2 \varrho)$
\LongVersionEnd %}
guaranteed by
\Thm{}~\ref{theorem_smoothness_of_polynomial_bounded_cost_sharing_mechanisms}
to obtain the following corollary.

\LongVersion \sloppy \LongVersionEnd
\begin{corollary*}
If $M$ is an REP-expanded CSM, then the
approximation ratio of \texttt{Alg-ABRD} is
$O \Left(
\varrho^{\max_{j} \alpha_{j} + 1}
+
\varrho \cdot \max_{e} \min_{j} \Left( \frac{\sigma_{e}}{\xi_{e, j}}
\Right)^{1 / \alpha_{j}}
\Right)$.
\end{corollary*}
\LongVersion \par\fussy \LongVersionEnd

\LongVersion %{
We now turn to show that some natural and extensively studied CSMs are
REP-expanded.
Under the \emph{proportional fair CSM} (see, e.g.,
\cite{Kollias2015RPE, Gkatzelis2014OCS}),
the cost share of player
$i \in [N]$
in edge
$e \in p_{i}$ is defined to be her share of the cost incurred by load
$l_{e}^{p}$ on edge $e$, proportional to her weight $w_{i}(e)$, namely 
$f_{i, e}(p) = \frac{w_{i}(e)}{l_{e}^{p}} F_{e}(l_{e}^{p})$.
\LongVersionEnd %}

\ShortVersion %{
\addtocounter{theorem}{1}
\ShortVersionEnd %}
\LongVersion %{
\begin{lemma}\label{lemma_fair_cost_sharing_is_polynomial_bounded}
The proportional fair CSM is REP-expanded.
\end{lemma}
\LongVersionEnd %}
\LongVersion %{
\begin{proof}
Under the proportional fair CSM, the cost share of player
$i$ in edge $e$ satisfies
\begin{align*}
\frac{w_{i}(e)}{l_{e}^{p}} \Left[ \sigma_{e} + \sum_{j = 1}^{q}\xi_{e,\,
j}(l_{e}^{p})^{\alpha_{j}} \Right]
\; \leq \; &
\sigma_{e} + \sum_{j = 1}^{q}\xi_{e,\, j}w_{i}(e) \cdot
(l_{e}^{p})^{\alpha_{j} - 1} \\
= \; &
\sigma_{e} + \sum_{j = 1}^{q}\xi_{e,\, j}w_{i}(e) \cdot \Left[ (l_{e}^{p} -
w_{i}(e)) + w_{i}(e) \Right]^{\alpha_{j} - 1} \, ,
\end{align*}
where the inequality holds because
$l_{e}^{p} = \sum_{i' : e \in p_{i'}} w_{i', e} \geq w_{i}(e)$.
Consider the following two cases.
\begin{DenseItemize}

\item
$0 \leq l_{e}^{p} - w_{i}(e) \leq w_{i}(e)$:
Since $\alpha - 1 > 0$, it follows that
\begin{align*}
\sigma_{e} + \sum_{j}\xi_{e,\, j}w_{i}(e) \cdot \Left[ (l_{e}^{p} - w_{i}(e))
+ w_{i}(e) \Right]^{\alpha_{j} - 1}
\; \leq \; &
\sigma_{e} + \sum_{j}\xi_{e,\, j}w_{i}(e) \cdot (2 \cdot w_{i}(e))^{\alpha_{j}
- 1} \\
= \; &
\sigma_{e} + \sum_{j}\xi_{e,\, j}2^{\alpha_{j} - 1}(w_{i}(e))^{\alpha_{j}} \, .
\end{align*}

\item
$l_{e}^{p} - w_{i}(e) > w_{i}(e)$:
In this case,
\[
\sigma_{e} + \sum_{j}\xi_{e,\, j} w_{i}(e) \cdot \Left[ (l_{e}^{p} - w_{i}(e)) + w_{i}(e) \Right]^{\alpha_{j}
- 1}
\; < \;
\sigma_{e} + \sum_{j}\xi_{e,\, j}2^{\alpha_{j} - 1} w_{i}(e)(l_{e}^{p} - w_{i}(e))^{\alpha_{j} - 1} \, .
\]

\end{DenseItemize}
The assertion follows by taking
$K_{j} = 2$,
$x_{1, j} = 0$,
$y_{1, j} = \alpha_{j}$,
$z_{1, j} = 2^{\alpha_{j} - 1}$,
$x_{2, j} = \alpha_{j} - 1$,
$y_{2, j} = 1$,
and
$z_{2, j} = 2^{\alpha_{j} - 1}$
for every $j \in [q]$.
\end{proof}
\LongVersionEnd %}

Let
$S_{e} = \{ i \in [N] \mid e \in p_{i} \}$
and let $\pi_{e}$ be a random permutation of $S_{e}$.
Under the \emph{Shapley CSM} (see, e.g.,
\cite{Kollias2015RPE, Gkatzelis2014OCS}),
the cost share of player
$i \in [N]$
in edge
$e \in p_{i}$
is defined to be its expected marginal contribution if the players are added
to $e$ one-by-one in $\pi_{e}$ order.
More formally, taking
$S_{e}^{i}(\pi_{e})$
to denote the set of players that precede player $i$ in $\pi_{e}$, the cost
share of $i$ in edge $e$ under the Shapley CSM is
\[
f_{i, e}(p)
\; = \;
\begin{cases}
\mathbb{E} \Left[
F_{e} \Left( \sum_{i' \in S_{e}^{i}(\pi_{e})}w_{i', e} + w_{i}(e) \Right) -
F_{e} \Left( \sum_{i' \in S_{e}^{i}(\pi_{e})}w_{i', e} \Right)
\Right] \, ,
& \text{if } e \in p_{i} \\
0 \, , & \text{otherwise}
\end{cases} \, .
\]

\begin{lemma} \label{lemma_shapley_cost_sharing_is_polynomial_bounded}
The Shapley CSM is REP-expanded.
\end{lemma}
\LongVersion %{
\begin{proof}
Under the Shapley CSM, the cost share of player
$i \in [N]$
in edge
$e \in p_{i}$
is
\begin{align*}
f_{i, e}(p)
\; \leq & \;
\mathbb{E} \Left[
\sigma_{e} +
\sum_{j}\xi_{e,\, j} \Left( \sum_{i' \in S_{e}^{i}(\pi_{e})} w_{i', e} + w_{i}(e) \Right)^{\alpha_{j}}
-
\sum_{j}\xi_{e,\, j} \Left( \sum_{i' \in S_{e}^{i}(\pi_{e})} w_{i', e} \Right)^{\alpha_{j}} \Right]
\\
\leq & \;
\sigma_{e} +
\sum_{j}\xi_{e,\, j} \Left( \sum_{i' \in (S_{e} - \lbrace i \rbrace)} w_{i', e} + w_{i}(e)
\Right)^{\alpha_{j}} -
\sum_{j}\xi_{e,\, j} \Left( \sum_{i' \in (S_{e} - \lbrace i \rbrace)} w_{i', e} \Right)^{\alpha_{j}}
\\
= & \;
\sigma_{e} + \sum_{j}\xi_{e,\, j}(l_{e}^{p})^{\alpha_{j}} - \sum_{j}\xi_{e,\, j} (l_{e}^{p} - w_{i}(e))^{\alpha_{j}} \,
,
\end{align*}
where
the first transition holds since for any $l \geq 0$,
$\sum_{j}\xi_{e,\, j} l^{\alpha_{j}} \leq F_{e}(l) \leq \sigma_{e} + \sum_{j}\xi_{e,\, j} l^{\alpha_{j}}$
and
the second transition holds since the function
$g(x) = x^{\alpha}$
is superadditive for any $\alpha > 1$.
Consider the following two cases.
\begin{DenseItemize}

\item
$l_{e}^{p} \leq 3 \cdot w_{i}(e)$:
In this case,
\begin{MathMaybe}
\sigma_{e} + \sum_{j}\xi_{e,\, j} (l_{e}^{p})^{\alpha_{j}} - \sum_{j}\xi_{e,\, j}(l_{e}^{p} - w_{i}(e))^{\alpha_{j}}
\; \leq \;
\sigma_{e} + \sum_{j}\xi_{e,\, j} \cdot 3^{\alpha_{j}} (w_{i}(e))^{\alpha_{j}} \, .
\end{MathMaybe}

\item
$l_{e}^{p} > 3 \cdot w_{i}(e)$:
By Newton's generalized binomial theorem, we have
\begin{align*}
&\; \sigma_{e} + \sum_{j}\xi_{e,\, j} (l_{e}^{p})^{\alpha_{j}} - \sum_{j}\xi_{e,\, j}(l_{e}^{p} - w_{i}(e))^{\alpha_{j}}\\
\leq & \;
\sigma_{e} + \sum_{j}\xi_{e,\, j} \sum_{k = 0}^{\infty} \binom{\alpha_{j}}{k} (l_{e}^{p} -
w_{i}(e))^{\alpha_{j} - k} (w_{i}(e))^{k} - \sum_{j}\xi_{e,\, j} (l_{e}^{p} - w_{i}(e))^{\alpha_{j}} \\
= & \;
\sigma_{e} + \sum_{j}\xi_{e,\, j} \sum_{k = 1}^{\infty} \binom{\alpha_{j}}{k} (l_{e}^{p} -
w_{i}(e))^{\alpha_{j} - k} (w_{i}(e))^{k} \\
< & \;
\sigma_{e} + \sum_{j}\xi_{e,\, j} \binom{\alpha_{j}}{\Left\lfloor \frac{\alpha_{j} + 1}{2} \Right\rfloor}
w_{i}(e) (l_{e}^{p} - w_{i}(e))^{\alpha_{j} - 1} \sum_{k = 1}^{\infty} \Left(
\frac{w_{i}(e)}{l_{e}^{p} - w_{i}(e)} \Right)^{k - 1} \\
< & \;
\sigma_{e} + \sum_{j}\xi_{e,\, j} \binom{\alpha_{j}}{\Left\lfloor \frac{\alpha_{j} + 1}{2}
\Right\rfloor} w_{i}(e) (l_{e}^{p} - w_{i}(e))^{\alpha_{j} - 1} \sum_{k = 1}^{\infty}
\Left( \frac{1}{2} \Right)^{k - 1} \\
< & \;
\sigma_{e} + 2 \sum_{j}\xi_{e,\, j} \binom{\alpha_{j}}{\Left\lfloor \frac{\alpha_{j} + 1}{2}
\Right\rfloor} w_{i}(e) (l_{e}^{p} - w_{i}(e))^{\alpha_{j} - 1} \, ,
\end{align*}
where
the third transition holds because for any
$\alpha > 1$
and
$k \in \mathbb{Z}^{+}$,
the absolute value of
$\binom{\alpha}{k}$
is at most
$\binom{\alpha}{\Left\lfloor \frac{\alpha + 1}{2} \Right\rfloor}$.

\end{DenseItemize}
The assertion follows by taking
$K_{j} = 2$,
$x_{1,\, j} = 0$,
$y_{1,\, j} = \alpha_{j}$,
$z_{1,\, j} = 3^{\alpha_{j}}$,
$x_{2,\, j} = \alpha_{j} - 1$,
$y_{2,\, j} = 1$,
and
$z_{2,\, j} = 2 \binom{\alpha_{j}}{\Left\lfloor \frac{\alpha_{j} + 1}{2} \Right\rfloor}$ 
for every $j \in [q]$.
\end{proof}
\LongVersionEnd %}

%%%%%%%%%%%%%%%%%%%%%%%%%%%%%%%%%%%%%%%%%%%%%%%%%%%%%%%%%%%%%%%%%%%%%%%%%%%%%%
\section{The Potential Function of the Shapley Cost Sharing Mechanism}
\label{sec:potent-funct-shapl}
%%%%%%%%%%%%%%%%%%%%%%%%%%%%%%%%%%%%%%%%%%%%%%%%%%%%%%%%%%%%%%%%%%%%%%%%%%%%%%
\LongVersion %{
The next step is to prove that among the REP-expanded CSMs, there exists one
that induces a GND game with an $(A, B)$-bounded potential function for
sufficiently small
$A, B \geq 1$.
(Recall that by
\Thm{}~\ref{theorem:approx-ratio-Alg-ABRD}, this would
provide an upper bound on the number of steps in the ABRD.)
While we could not accomplish this task for the proportional fair CSM, the
Shapley CSM turned out to be more successful.
\LongVersionEnd %}
\ShortVersion %{
The next step is to prove that the Shapley CSM is not only REP-expanded, but
it also induces a GND game with an $(A, B)$-bounded potential function for
sufficiently small
$A, B \geq 1$
that when plugged into \Thm{}~\ref{theorem:approx-ratio-Alg-ABRD},
provide an upper bound on the number of steps in the ABRD.
Following that, there is still some work to be done to establish the desired
strongly polynomial run-time bound, but this is deferred to the attached full
version.
\ShortVersionEnd %}

It can be inferred from \cite{Roughgarden2016NCS, Hart1989PVC, Kollias2015RPE}
that the GND game induced by the Shapley CSM admits the
potential function
\begin{equation}
\label{formula_potential_function_guaranteed_by_Shapley_cost_sharing}
\Phi(p)
\; = \;
\sum_{e \in p} \sum_{i \in S_{e}} f_{i, e}(S_{e}^{i}(\psi_{e}) \cup \lbrace
i \rbrace )\, ,
\end{equation}
where
$S_{e} = \{ j \in [N] \mid e \in p_{j} \}$,
$\psi_{e}$ is an arbitrary permutation of $S_{e}$,
and
$S_{e}^{i}(\psi_{e})$
is the set of players that precede $i$ in $\psi_{e}$.
Note that in contrast to the random permutation $\pi_{e}$ used in the
definition of the Shapley CSM, the permutation $\psi_{e}$
is an (arbitrary) deterministic permutation.
The rest of this section is dedicated to proving that this potential function
is
$(\mathcal{H}_{N}, \lceil \max_{j}\alpha_{j} \rceil)$-bounded,
where $\mathcal{H}_{N}$ is the $N$-th harmonic number.
Define the function
$h_{e}: 2^{[N]} \rightarrow \mathbb{R}_{\geq 0}$
by setting
\begin{MathMaybe}
h_{e}(X)
\, = \,
\sum_{j \in [q]}
\xi_{e,\, j} \cdot \Left( \sum_{i \in X} w_{i}(e) \Right)^{\alpha_{j}} \, .
\end{MathMaybe}

\ShortVersion %{
\addtocounter{theorem}{1}
\ShortVersionEnd %}
\LongVersion %{
\begin{lemma}
\label{lemma_transform_the_potential_function_of_Shapley_cost_sharing}
For any edge $e$ and any permutation $\psi_{e}$, we have
\[
\sum_{i \in S_{e}}
f_{i, e}(S_{e}^{i}(\psi_{e}) \cup \lbrace i \rbrace)
\, = \,
\sum_{k = 1}^{|S_{e}|}
\Left(
\frac{\sigma_{e}}{k} +
\sum_{T \subseteq S_{e}, |T| = k}
\frac{h_{e}(T)}{\binom{|S_{e}|}{k} \cdot k}
\Right) \, .
\]
\end{lemma}
\LongVersionEnd %}
\LongVersion %{
\begin{proof}
By the definition of the Shapley CSM, the cost share of
player $i$ who uses edge $e$ is
\[
f_{i, e}(S_{e})
\; = \;
\mathbb{E} \Left[
F_{e} \Left( \sum_{i' \in S_{e}^{i}(\pi_{e})} w_{i', e} + w_{i}(e) \Right) -
F_{e} \Left( \sum_{i' \in S_{e}^{i}(\pi_{e})} w_{i', e} \Right)
\Right] \, ,
\]
where $\pi_{e}$ is a random permutation on $S_{e}$.
For a fixed $\pi_{e}$ and a fixed player $i$ using edge $e$,
\begin{align*}
& F_{e} \Left( \sum_{i' \in S_{e}^{i}(\pi_{e})} w_{i', e} + w_{i}(e) \Right) -
F_{e} \Left( \sum_{i' \in S_{e}^{i}(\pi_{e})}w_{i', e} \Right) \\
= \, &
\begin{cases}
\sigma_{e} + h_{e}(\lbrace i \rbrace) - h_{e}(\emptyset) \, , & \text{if }
S_{e}^{i}(\pi_{e}) = \emptyset \\
h_{e}(S_{e}^{i}(\pi_{e}) \cup \lbrace i \rbrace) - h_{e}(S_{e}^{i}(\pi_{e}))
\, , & \text{otherwise}
\end{cases} \\
= \, &
1(S_{e}^{i}(\pi_{e}) = \emptyset) \cdot \sigma_{e} +
h_{e}(S_{e}^{i}(\pi_{e}) \cup \lbrace i \rbrace) - h_{e}(S_{e}^{i}(\pi_{e}))
\, ,
\end{align*}
where
$1(S_{e}^{i}(\pi_{e}) = \emptyset)$
denotes the indicator of the event
$S_{e}^{i}(\pi_{e}) = \emptyset$.
Since $\pi_{e}$ is taken from the uniform distribution, it follows that
$\mathbb{P}(S_{e}^{i}(\pi_{e}) = \emptyset) = \frac{1}{|S_{e}|}$ for any
player $i$ using edge $e$, thus
\begin{equation}
\label{formula_transform_of_the_expectation_in_Shapley_value}
\mathbb{E} \Left[
F_{e} \Left( \sum_{i' \in S_{e}^{i}(\pi_{e})} w_{i', e} + w_{i}(e) \Right) -
F_{e} \Left( \sum_{i' \in S_{e}^{i}(\pi_{e})} w_{i', e} \Right)
\Right]
\; = \;
\frac{\sigma_{e}}{|S_{e}|} +
\mathbb{E} \Left[
h_{e}( S_{e}^{i}(\pi_{e}) \cup \lbrace i \rbrace ) - h_{e}(S_{e}^{i}(\pi_{e}))
\Right] \, .
\end{equation}

Let
$H_{i, e}(S_{e})
=
\mathbb{E}[
h_{e}( S_{e}^{i}(\pi_{e}) \cup \lbrace i \rbrace ) - h_{e}(S_{e}^{i}(\pi_{e}))]$.
Then, it can be inferred from
Eq.~(\ref{formula_transform_of_the_expectation_in_Shapley_value}) that for any
player $i$ using $e$,
\[
f_{i, e}(S_{e}^{i}(\psi_{e}) \cup \lbrace i \rbrace )
\; = \;
\frac{\sigma_{e}}{|S_{e}^{i}(\psi_{e}) \cup \lbrace i \rbrace|} +
H_{i, e} \Left( S_{e}^{i}(\psi_{e}) \cup \lbrace i \rbrace \Right) \, .
\]
Since
$\sum_{i \in S_{e}} \frac{\sigma_{e}}{|S_{e}^{i}(\psi_{e}) \cup \lbrace i
\rbrace|}
=
\sum_{k = 1}^{|S_{e}|}\frac{\sigma_{e}}{k}$,
it follows that
\[
\sum_{i \in S_{e}} f_{i, e}(S_{e}^{i}(\psi_{e}) \cup \lbrace i \rbrace )
\; = \;
\sum_{k = 1}^{|S_{e}|} \frac{\sigma_{e}}{k} +
\sum_{i \in S_{e}} H_{i, e} \Left( S_{e}^{i}(\psi_{e}) \cup \lbrace i
\rbrace \Right) \, .
\]

Notice that $H_{i, e}(S_{e})$ can be viewed as the Shapley cost share of a
player $i$ who uses edge $e$ in a network game where the cost of each edge $e$
is $h_{e}(S_{e})$.
Since for every $j$, $\alpha_{j}$ is assumed to be larger than $1$, the ratio
$\frac{h(S_{e})}{l_{e}} = \frac{\sum_{j}\xi_{e,\, j}(l_{e})^{\alpha_{j}}}{l_{e}}$
is non-decreasing with $l_{e}$.
Kollias and Roughgarden \cite{Kollias2015RPE} prove that in such case, for any
$\psi_{e}$,
\[
\sum_{i \in S_{e}} H_{i, e} \Left( S_{e}^{i}(\psi_{e}) \cup \lbrace i
\rbrace \Right)
\; = \;
\sum_{T \subseteq S_{e}, |T| = k} \frac{h_{e}(T)}{\binom{|S_{e}|}{k}
\cdot k} \, ,
\]
thus establishing the assertion.
\end{proof}
\LongVersionEnd %}

\LongVersion %{
We are now ready to prove that the potential function $\Phi(p)$ of
(\ref{formula_potential_function_guaranteed_by_Shapley_cost_sharing}) is
$(\lceil \max_{j} \alpha_{j} \rceil, \mathcal{H}_{N})$-bounded.
\LongVersionEnd %}

\begin{theorem}
\label{theorem_bounds_on_the_potential_function_of_Shapley_cost_sharing}
The potential function $\Phi(p)$ of the GND game induced by the Shapley CSM
satisfies
$\frac{1}{\lceil \max_{j} \alpha_{j} \rceil} \cdot C(p)
\leq
\Phi(p)
\leq
\mathcal{H}_{N} \cdot C(p)$
for any strategy profile $p$.
\end{theorem}
\begin{proof}
\LongVersion %{
Let us first prove the lower bound on $\Phi(p)$.
\LongVersionEnd %}
\ShortVersion %{
We establish the lower bound on $\Phi(p)$ here and defer the upper bound to
the attached full version.
\ShortVersionEnd %}
Since
$e \in p$
implies that
$|S_{e}| \geq 1$,
we get
\[
\Phi(p)
\, = \,
\sum_{e \in p} \sum_{k = 1}^{|S_{e}|}
\LLeft(
\frac{\sigma_{e}}{k} +
\sum_{T \subseteq S_{e}, |T| = k}
\frac{h_{e}(T)}{\binom{|S_{e}|}{k} \cdot k}
\RRight)
\, \geq \,
\sum_{e \in p}
\LLeft(
\sigma_{e} +
\sum_{k = 1}^{|S_{e}|} \sum_{T \subseteq S_{e}, |T| = k}
\frac{h_{e}(T)}{\binom{|S_{e}|}{k} \cdot k}
\RRight) \, .
\]
By the convexity of
$\xi_{e,\, j} \cdot x^{\alpha_{j}}$,
we conclude that
\[
\Phi(p)
\, \geq \,
\sum_{e \in p}
\LLeft[
\sigma_{e} +
\sum_{j = 1}^{q}
\xi_{e,\, j} \sum_{k = 1}^{|S_{e}|}
\frac{1}{\binom{|S_{e}|}{k} \cdot k} \binom{|S_{e}|}{k}
\Left(
\frac{\sum_{T \subseteq S_{e}, |T| = k} \sum_{i \in T} w_{i}(e)}{\binom{|S_{e}|}{k}}
\Right)^{\alpha_{j}}
\RRight] \, .
\]
Since every player $i$ is included in exactly
$\binom{|S_{e}| -  1}{k - 1}$
subsets of $S_{e}$ with $k$ elements, it follows that
\begin{align*}
\Phi(p)
\, \geq \, &
\sum_{e \in p}
\LLeft[
\sigma_{e} +
\sum_{j = 1}^{q}
\xi_{e, j} \sum_{k = 1}^{|S_{e}|} \frac{1}{k}
\Left(
\frac{\binom{|S_{e}| - 1}{k - 1} l_{e}^{p}}{\binom{|S_{e}|}{k}}
\Right)^{\alpha_{j}}
\RRight]
\, \geq \,
\sum_{e \in p}
\LLeft[
\sigma_{e} +
\sum_{j = 1}^{q} \xi_{e, j}
\Left(
\frac{l_{e}^{p}}{|S_{e}|}
\Right)^{\alpha_{j}}
\sum_{k = 1}^{|S_{e}|} k^{\alpha_{j} - 1}
\RRight] \\
\geq \, &
\sum_{e \in p}
\LLeft[
\sigma_{e} + \sum_{j = 1}^{q}
\xi_{e, j}
\Left(
\frac{l_{e}^{p}}{|S_{e}|}
\Right)^{\alpha_{j}}
\Left(
\frac{1}{|S_{e}|}
\Right)^{\lceil \alpha_{j} \rceil - \alpha_{j}}
\sum_{k = 1}^{|S_{e}|} k^{\lceil \alpha_{j} \rceil - 1}
\RRight] \, .
\end{align*}
Employing \cite{Anshelevich2008PSN}, we conclude that
\begin{MathMaybe}
\Phi(p)
\, \geq \,
\sum_{e \in p}
\LLeft[
\sigma_{e} +
\sum_{j = 1}^{q}
\xi_{e, j}
\Left(
\frac{l_{e}^{p}}{|S_{e}|}
\Right)^{\alpha_{j}}
\Left(
\frac{1}{|S_{e}|}
\Right)^{\lceil \alpha_{j} \rceil - \alpha_{j}}
\frac{|S_{e}|^{\lceil \alpha_{j} \rceil}}{\lceil \alpha_{j} \rceil}
\RRight]
\, \geq \,
\frac{1}{\lceil \max_{j}\alpha_{j} \rceil}
\sum_{e \in p}
[\sigma_{e} + \sum_{j}\xi_{e, j}(l_{e}^{p})^{\alpha_{j}}] \, .
\end{MathMaybe}
\ShortVersion %{
The assertion follows.
\ShortVersionEnd %}
\LongVersion %{
\par
For the upper bound on $\Phi(p)$, by
\Lem{}~\ref{lemma_transform_the_potential_function_of_Shapley_cost_sharing}
and since
$h_{e}(T)$ is a (set-wise) increasing function of $T$ and
$T \subseteq S_{e}$,
we get
\[
\Phi(p)
\, = \,
\sum_{e \in p} \sum_{k = 1}^{|S_{e}|}
\Left(
\frac{\sigma_{e}}{k} +
\sum_{T \subseteq S_{e}, |T| = k}
\frac{h_{e}(T)}{\binom{|S_{e}|}{k} \cdot k}
\Right)
\, \leq \,
\sum_{e \in p}
\Left(
\mathcal{H}_{|S_{e}|} \cdot \sigma_{e} +
\sum_{k = 1}^{|S_{e}|} \sum_{T \subseteq S_{e},  |T| = k}
\frac{h_{e}(S_{e})}{\binom{|S_{e}|}{k} \cdot k}
\Right) \, .
\]
As there are exactly
$\binom{|S_{e}|}{k}$
subsets of $S_{e}$ with $k$ elements, it follows that
\begin{align*}
\Phi(p)
\, \leq \,
\sum_{e \in p} \Left( \mathcal{H}_{|S_{e}|} \cdot \sigma_{e} + \sum_{k =
1}^{|S_{e}|} \frac{h_{e}(S_{e})}{k} \Right)
\, = \, &
\sum_{e \in p} \Left( \mathcal{H}_{|S_{e}|} \cdot \sigma_{e} +
\mathcal{H}_{|S_{e}|} \cdot h_{e}(S_{e}) \Right) \\
= \, &
\sum_{e \in p} \mathcal{H}_{|S_{e}|} \cdot \Left[\sigma_{e} + \sum_{j \in
[q]}\xi_{e,\, j} \Left( \sum_{i \in S_{e}} w_{i}(e) \Right)^{\alpha_{j}}
\Right] \, .
\end{align*}
Since $e \in p$
means that
$l_{e}^{p} > 0$, we conclude that
\[
\Phi(p)
\, \leq \,
\sum_{e \in p} \mathcal{H}_{|S_{e}|} F_{e}(l_{e}^{p})
\, \leq \,
\mathcal{H}_{N} \cdot C(p)
\]
which establishes the assertion.
\LongVersionEnd %}
\end{proof}

\LongVersion %{
%%%%%%%%%%%%%%%%%%%%%%%%%%%%%%%%%%%%%%%%%%%%%%%%%%%%%%%%%%%%%%%%%%%%%%%%%%%%%%
\section{Polynomial-Time $\epsilon$-Approximation of Shapley Cost Sharing}
\label{sec:polyn-time-epsil}
%%%%%%%%%%%%%%%%%%%%%%%%%%%%%%%%%%%%%%%%%%%%%%%%%%%%%%%%%%%%%%%%%%%%%%%%%%%%%%
So far, we have proved that the Shapley CSM can satisfy the
requirements on the smoothness and the potential function.
It remains to show how to compute an $\epsilon$-approximation of the cost
shares specified by the Shapley CSM in polynomial time
for a sufficiently small $\epsilon > 0$.
For the problem of computing the cost shares specified by the Shapley CSM,
Liben-Nowell et al.~\cite{Liben2012CSV} establish the following lemma.

\begin{lemma}[\cite{Liben2012CSV}]
\label{lemma_Liben_SV_sampling_algorithm_for_supermodular_game}
There exists an FPRAS (i.e., a randomized FPTAS), referred to as
\texttt{SV-Sample}, for computing the cost shares in any game subject to the
Shapley CSM and \emph{supermodular} monotone cost
functions.
In particular, given any
$\varepsilon \in (0, 1)$,
\texttt{SV-Sample} generates an $\varepsilon$-approximation of the cost share
with probability at least $1 - \frac{1}{2(TN|E|)^{2}}$ in
$O\Big(\log(TN|E|) \cdot \Big[\frac{N^{3}}{\varepsilon^{2}} + \log(\log(TN|E|)) \Big] \Big)$-time.
\end{lemma}

Note that owing to the existence of the term $\sigma_{e}$, the cost function
$F_{e}(l)$ is not supermodular.
The following lemma comes to our help.

\begin{lemma}\label{lemma_supermodularity_of_polynomial_function}
The function $h_{e}(X)$ (defined in the statement of \Lem{}~\ref{lemma_transform_the_potential_function_of_Shapley_cost_sharing}) is supermodular.
\end{lemma}
\begin{proof}
Recall that
$h_{e}(X) =
\sum_{j \in [q]} \xi_{e, j} \left(\sum_{i \in X} w_{i}(e) \right)^{\alpha}$.
By definition, $h_{e}(X)$ is supermodular if for any $X_{1} \subset X_{2} \subset X$ and $i' \in X \backslash X_{2}$:
\begin{equation}
\sum_{j = 1}^{q}\xi_{e,\, j}\Big[\Big(\sum_{i \in X_{1} \cup \lbrace i' \rbrace} w_{i}(e) \Big)^{\alpha_{j}} - \Big(\sum_{i \in X_{1} } w_{i}(e) \Big)^{\alpha_{j}}\Big]
 \; \leq \; 
\sum_{j = 1}^{q}\xi_{e,\, j}\Big[ \Big(\sum_{i \in X_{2} \cup \lbrace i' \rbrace} w_{i}(e) \Big)^{\alpha_{j}} - \Big(\sum_{i \in X_{2} } w_{i}(e) \Big)^{\alpha_{j}}\Big] \nonumber
\end{equation}
Therefore, it suffices to prove that for any three non-negative numbers $x_{1}$, $x_{2}$ and $y$, $(x_{1} + y)^{\alpha_{j}} - (x_{1})^{\alpha_{j}} \leq (x_{1} + x_{2} + y)^{\alpha_{j}} - (x_{1} + x_{2})^{\alpha_{j}}$ holds for any $\alpha_{j} > 1$. Let $v(x_{1}, x_{2}, y) = (x_{1} + x_{2} + y)^{\alpha_{j}} - (x_{1} + x_{2})^{\alpha_{j}}$. Then we have:
\begin{equation}
\frac{\partial v}{\partial x_{2}} \; = \; \alpha_{j}\Big[ (x_{1} + x_{2} + y)^{\alpha_{j} - 1} - (x_{1} + x_{2})^{\alpha_{j} - 1} \Big] \nonumber
\end{equation}
Since $\alpha_{j} > 1$, $\frac{\partial v}{\partial x_{2}} > 0$ when $x_{2} > 0$. Therefore, $v(x_{1}, x_{2}, y) \geq v(x_{1}, 0, y) = (x_{1} + y)^{\alpha_{j}} - (x_{1})^{\alpha_{j}}$.
\end{proof}

Combining Lemma \ref{lemma_Liben_SV_sampling_algorithm_for_supermodular_game} and Lemma \ref{lemma_supermodularity_of_polynomial_function} with Eq.~\eqref{formula_transform_of_the_expectation_in_Shapley_value}, we get an efficient procedure, named \texttt{Shapley-APX}, for computing the $\epsilon$-approximation of the Shapley cost share for any given player $i$ and edge $e$ with respect to our cost function.

More specifically, if $i \notin S_{e}$, then this procedure returns $0$ as the cost share. Otherwise, \texttt{Shapley-APX} uses algorithm \texttt{SV-Sample} to obtain an $\epsilon$-approximation $\theta_{i, e}$ of the Shapley cost share for player $i$ on edge $e$ with respect to the cost function $h_{e}(S_{e})$. Finally, $\frac{\sigma_{e}}{|S_{e}|} + \theta_{i, e}$ is returned as the desired $\epsilon$-cost share. By Eq.~\eqref{formula_transform_of_the_expectation_in_Shapley_value}, Lemma \ref{lemma_Liben_SV_sampling_algorithm_for_supermodular_game} and Lemma \ref{lemma_supermodularity_of_polynomial_function}, the following lemma trivially holds.

\begin{lemma}\label{lemma_approximating_Shapley_cost_share_with_high_probability_in_polynomial_time}
Procedure \texttt{Shapley-APX} computes an $\epsilon$-approximation of the cost share of a player $i$ on edge $e$ with probability at least $1 - \frac{1}{2(TN|E|)^2}$ in $O\Big(\log(TN|E|) \cdot \Big[\frac{N^{3}}{\epsilon^{2}} + \log(\log(TN|E|)) \Big] \Big)$-time.
\end{lemma}

\begin{theorem}
\label{theorem_the_probabilty_that_we_do_not_make_any_mistake_in_main_algorithm}
If \texttt{Shapley-APX} is employed to generate all the $\epsilon$-cost
shares used in \texttt{Alg-ABRD}, then w.h.p., every $\epsilon$-cost share
$\widetilde{f}_{i, e}(S_{e})$
is an $\epsilon$-approximation of the exact cost share $f_{i, e}(S_{e})$.
\end{theorem}

\begin{proof}
Recall that the ABRD contains at most $T$ steps. In each step $t$, every player $i$ needs to calculate $\widetilde{f}_{i, e}(S_{e}^{t} \cup \lbrace i \rbrace)$ for every edge $e$ to find its best approximate response. Therefore, procedure \texttt{Shapley-APX} is invoked at most $TN|E|$ times. The probability that this function generates a result that is not the $\epsilon$-approximation of the exact cost share is at most $\displaystyle TN|E| \cdot \Big[1 - \Big(1 - \frac{1}{2(TN|E|)^2} \Big) \Big] = \frac{1}{2TN|E|}$. Therefore, this theorem follows. 
\end{proof}

Using the facts that $\epsilon$, $\lambda_{\alpha}$ (from
\Thm{}~\ref{theorem_smoothness_of_polynomial_bounded_cost_sharing_mechanisms}),
and
$\varrho \epsilon_{1}^{2} / 2$
are all constants, and $\mathcal{H}_{N}$ can be bounded by $O(\log N)$, the main result is established as summed up in the following theorem.

\begin{theorem}
\label{theorem_formal_description_of_the_main_result}
By plugging the Shapley CSM into \texttt{Alg-ABRD}, the total cost of the
output profile is an
$O\left(
\varrho \cdot \max_{e}\min_{j} \Big( \frac{\sigma_{e}}{\xi_{e, j}}
\Big)^{\frac{1}{\alpha_{j}}} + \varrho^{\max_{j}\alpha_{j} + 1}
\right)$-approximation of the optimal result with probability at least
\[
1 - O\left( \frac{1}{N^{2}|E|\log^{2}N} \right) \, .
\]
The time complexity of the algorithm is
\[
O \bigg( N^{5} \log^{2} N \cdot |E|\log^{2}(N|E|) \bigg) \, .
\]
\end{theorem}
\LongVersionEnd %}

\LongVersion %{
%%%%%%%%%%%%%%%%%%%%%%%%%%%%%%%%%%%%%%%%%%%%%%%%%%%%%%%%%%%%%%%%%%%%%%%%%%%%%%
\section{Implementation in a Decentralized Environment}
\label{sec:impl-decentr-envir}
%%%%%%%%%%%%%%%%%%%%%%%%%%%%%%%%%%%%%%%%%%%%%%%%%%%%%%%%%%%%%%%%%%%%%%%%%%%%%%
The approximation algorithm \texttt{Alg-ABRD} that was developed up to now is centralized, and in particular two main aspects of the algorithm are incompatible with some common settings in game theory. The first aspect is that \texttt{Alg-ABRD} deterministically chooses a specific player for strategy update. Instead, if traffic requests were separate uncoordinated entities, it would make more sense that they decide to update their strategies in an \emph{uncoordinated} way. The second aspect is that \texttt{Alg-ABRD} chooses the best profile it has seen during the ABRD. However, it is inappropriate in game theory to ask uncoordinated individual entities to ``roll back'' to a previous profile that might be more costly for some of them.
 
This section tackles these issues by providing two techniques for adapting algorithm \texttt{Alg-ABRD} to the game-theoretic settings. First, instead of choosing a specific player for updating the strategy, we now select the player uniformly at random. We believe that this better simulates the behaviors of uncoordinated players. Subsection \ref{subsection_randomized_selection_and_decentralized_implementation} shows that this modification will still yield the same approximation ratio, with only a polynomial loss in the number of steps. Second, instead of choosing the best configuration in the sequence, subsection \ref{subsection_no_roll_back} analyzes the case where the last configuration is chosen. It is shown that the approximation ratio loses another $O(\log N)$ factor. Thus, while certainly inferior to the centralized algorithm, the game-theoretic version of \texttt{Alg-ABRD} still admits an approximately optimal outcome.

%%%%%%%%%%%%%%%%%%%%%%%%%%%%%%%%%%%%%%%
\subsection{Randomized Selection and Decentralized Implementation}
\label{subsection_randomized_selection_and_decentralized_implementation}
%%%%%%%%%%%%%%%%%%%%%%%%%%%%%%%%%%%%%%%
This subsection develops a random procedure of selecting the players for strategy updates. Specifically, it is assumed that all the players share the same source of random bits, which generates a number $i \in [N]$ randomly and uniformly for each step $t$ in which the ABRD has not converged. If the player with the same index satisfies:
\begin{equation}
\delta_{i}^{t} \; > \; 0 \label{formula_condition_for_updating_in_randomized_selection_of_A_BRD}
\end{equation}
then it unilaterally deviates to its best approximate response $p_{i}'$ to $p^{t}$. Otherwise, it does nothing. %Different from the technique of deterministic selection used in \texttt{BARD-APX}, now player $i$ can determine whether to update its strategy independently of any global variable which requires centralized computation. Therefore, this \emph{random selection} can be efficiently implemented in the distributed setting.

\begin{lemma}\label{lemma_potential_function_decreases_in_each_step_of_ABRD}
For any step $t \geq 1$ where the ABRD does not converge, as long as the player selected for strategy update satisfies Eq.~(\ref{formula_condition_for_updating_in_randomized_selection_of_A_BRD}), we have $\Phi(p^{t - 1}) - \Phi(p^{t}) > 0$.
\end{lemma}

\begin{proof}
Since the ABRD does not converge at step $t$, there exists a player $j$ who is selected to update its strategy. By the definition of the potential function:
\begin{align}
\Phi(p^{t - 1}) - \Phi(p^{t}) \; =& \; C_{j}(p^{t - 1}) -
C_{j}(p_{j}^{t}, p_{-j}^{t - 1}) \nonumber\\
	\geq& \; \frac{1}{1 + \epsilon}\widetilde{C}_{j}(p^{t - 1}) - \frac{1}{1 - \epsilon}\widetilde{C}_{j}(p_{j}^{t}, p_{-j}^{t - 1}) \nonumber\\
	>& \; \frac{1}{1 + \epsilon} \epsilon_{1}
\widetilde{C}_{j}(p_{j}^{t}, p_{-j}^{t - 1}) - \frac{1}{1 -
\epsilon}\widetilde{C}_{j}(p_{j}^{t}, p_{-j}^{t - 1}) \nonumber\\
	=& \; 0 \nonumber
\end{align}
The second formula follows from the definition of the $\epsilon$-individual cost. The third one follows from Eq.~(\ref{formula_condition_for_updating_in_randomized_selection_of_A_BRD}).
\end{proof}

\begin{theorem}\label{theorem:approx-ratio-Alg-ABRD_with_randomized_selection}
With probability at least $\frac{1}{2}$, the output $p^{t^{*}}$ of the revised algorithm which uses the random selection can satisfy:
\begin{equation}
C(p^{t^{*}}) \; \leq \; \frac{2\varrho \epsilon_{1}^2\lambda}{1 - \varrho \epsilon_{1}^2\mu} \cdot C^{*} \nonumber
\end{equation}
if we change the number of steps in ABRD to $\displaystyle T' = N \cdot T^{2}$.
\end{theorem}

\begin{proof}
According to Lemma \ref{lemma:ABRD-converge-approx-NE}, here we can still only focus on the case where the ABRD does not converge. In this case, the $T'$ steps need to be partitioned into $T$ \emph{stages}, each of which contains $N \cdot T$ steps. We say a player $i$ is \emph{appropriate} for step $t$ if $\delta_{j}^{t}$ satisfies Eq.~(\ref{formula_condition_of_being_selected_for_updating_the_strategy}). A step $t$ is said to be appropriate if in this step an appropriate player is selected, and a stage is appropriate if it contains at least one appropriate step. Then we have:

\begin{claim}\label{proposition_with_random_selection_all_stages_are_appropriate}
With probability at least $\frac{1}{2}$, all the $T$ stages are appropriate.
\end{claim}

\begin{subproof}
\qedlabel{proposition_with_random_selection_all_stages_are_appropriate}
If the ABRD does not converge, then the Pigeonhole Principle shows that there exists at least one appropriate player in each step. By the assumption on the random source, each step is appropriate with probability $\frac{1}{N}$. Therefore, the probability that there exists no appropriate step in a stage must be:
\begin{equation}
\Big(1 - \frac{1}{N} \Big)^{NT} \; \leq \; \Big( \frac{1}{\operatorname{exp}(1)} \Big)^{T} \nonumber
\end{equation}

\sloppy
Hence, the probability that all the stages are appropriate is:
\begin{equation}
\Left( 1 - \Big(1 - \frac{1}{N} \Big)^{NT} \Right)^{T} \; \geq \; \Left( 1 - \Big( \frac{1}{\operatorname{exp}(1)} \Big)^{T} \Right)^{T} \; > \; \Left( 1 - \Big( \frac{1}{2} \Big)^{T} \Right)^{T} > \frac{1}{2} \nonumber
\end{equation}
The last inequality above holds because $\Big( \frac{1}{2} \Big)^{T} + \Big( \frac{1}{2} \Big)^{1 / T}$ decreases with $T$ when $T > 1$.
\end{subproof}
\par\fussy

Recall that we say a profile $p^{t}$ is bad if $C(p^{t}) \; > \;
\frac{2 \varrho \epsilon_{1}^{2} \lambda}{1 - \varrho \epsilon_{1}^{2} \mu}
\cdot C^{*}$.
Then we have:

\begin{claim}\label{proposition_with_random_selection_potential_function_decreases_significantly_in_appropriate_stages}
For any appropriate stage $\Bbbk$, if all the profiles generated in this stage are bad, then $\displaystyle \Phi(p^{t_{\Bbbk}^{1}}) < \Big(1 - \frac{1}{Q} \Big)\Phi(p^{t_{\Bbbk}^{0} - 1})$, where $t_{\Bbbk}^{0}$ an $t_{\Bbbk}^{1}$ respectively represent the first and the last steps in stage $\Bbbk$.
\end{claim}

\begin{subproof}
\qedlabel{proposition_with_random_selection_potential_function_decreases_significantly_in_appropriate_stages}
Since the player selected in each step does not update its strategy unless Eq.~(\ref{formula_condition_for_updating_in_randomized_selection_of_A_BRD}) is satisfied, Lemma \ref{lemma_potential_function_decreases_in_each_step_of_ABRD} indicates that the potential function is non-increasing in all the steps. Let $t_{\Bbbk}^{*}$ be an arbitrary appropriate step in stage $\Bbbk$. Then we have:
\begin{equation}
\Phi(p^{t_{\Bbbk}^{1}}) \; \leq \; \Phi(p^{t_{\Bbbk}^{*}}) \; < \; \Big(1 - \frac{1}{Q} \Big)\Phi(p^{t_{\Bbbk}^{*} - 1}) \; \leq \; \Big(1 - \frac{1}{Q} \Big)\Phi(p^{t_{\Bbbk}^{0} - 1}) \, , \nonumber
\end{equation}
where the second inequality follows from Claim \ref{proposition_potential_function_decreases_significantly_in_the_steps_with_bad_states}.
\end{subproof}

Combining Claim \ref{proposition_with_random_selection_potential_function_decreases_significantly_in_appropriate_stages} with the techniques in the proof of Theorem \ref{theorem:approx-ratio-Alg-ABRD}, we can prove that if all the stages are appropriate, then at least one stage generates a profile that is not bad.
\end{proof}

Putting Theorem \ref{theorem:approx-ratio-Alg-ABRD_with_randomized_selection}, Theorem \ref{theorem_smoothness_of_polynomial_bounded_cost_sharing_mechanisms}, Theorem \ref{theorem_bounds_on_the_potential_function_of_Shapley_cost_sharing} and Theorem \ref{theorem_the_probabilty_that_we_do_not_make_any_mistake_in_main_algorithm} together, we can get:

\begin{corollary}\label{corollary_result_for_random_selction}
The revised algorithm where the player for strategy update is determined uniformly at random yields the same approximation ratio as \texttt{Alg-ABRD} with probability at least
\[
\frac{1}{2}\Left(1 - O \Left( \frac{1}{N^{2}|E|\log^{2}N} \Right) \Right) \, .
\]
The time complexity of the revised algorithm is
\[
O \bigg( N^{7} \log^{4} N \cdot |E| \log^{2}(N|E|) \bigg) \, .
\]
\end{corollary}

%%%%%%%%%%%%%%%%%%%%%%%%%%%%%%%%%%%%%%%
\subsection{Output the Last Strategy Profile}
\label{subsection_no_roll_back}
%%%%%%%%%%%%%%%%%%%%%%%%%%%%%%%%%%%%%%%
In this subsection, we study the approach of directly returning the last strategy instead of the one with the minimum overall cost, $p^{t^{*}}$. 

\begin{theorem}\label{theorem_gap_between_the_costs_of_the_best_profile_and_last_one}
Suppose that the ABRD does not converge at any step $t$, then the cost corresponding to the last profile is at most $\lceil \max_{j}\alpha_{j} \rceil\mathcal{H}_{N}$ times larger than $C(p^{t^{*}})$, no matter whether the player in each step is selected in the deterministic way or in a random one.
\end{theorem}

\begin{proof}
Let $t_{\max}$ be the last step of the ABRD. 
This theorem trivially holds when $t^{*} = t_{\max}$. Suppose that $t^{*} < t_{\max}$, then:
\begin{equation}
C(p^{t_{\max}}) \; \leq \; \lceil \max_{j}\alpha_{j} \rceil \Phi(p^{t_{\max}}) \; < \; \lceil \max_{j}\alpha_{j} \rceil \Phi(p^{t^{*}}) \; \leq \; \lceil \max_{j}\alpha_{j} \rceil\mathcal{H}_{N} C(p^{t^{*}}) \nonumber
\end{equation}
The first inequality and the last one hold since the potential function is $(\mathcal{H}_{N}, \lceil \alpha \rceil)$-bounded. The second inequality follows from Eq.~(\ref{formula_condition_of_being_selected_for_updating_the_strategy}), Eq.~(\ref{formula_condition_for_updating_in_randomized_selection_of_A_BRD}) and Lemma \ref{lemma_potential_function_decreases_in_each_step_of_ABRD}, for both the deterministic selection and the random selection. %holds for the following reason. Since $t^{*} < t_{\max}$, the BARD does not converge at step $t^{*}$. By Eq.~(\ref{formula_condition_of_being_selected_for_updating_the_strategy}), Eq.~(\ref{formula_condition_for_updating_in_randomized_selection_of_A_BRD}) and Lemma \ref{lemma_potential_function_decreases_in_each_step_of_ABRD}, the potential function $\Phi$ is non-increasing for any step $t \in \lbrace t^{*}, \cdots, t_{\max} - 1 \rbrace$, no matter the player in each step is selected in a deterministic way or a random one.
\end{proof}

Since $\mathcal{H}_{N}$ is bounded by $O(\log N)$, it can be inferred that:

\begin{corollary}\label{corollary_loss_in_approximation_caused_by_non_roll_back}
Returning the strategy profile generated in the last step of the ABRD instead of $p^{t^{*}}$ as the output will increase the approximation ratio by $O(\log N)$ times.
\end{corollary}
\LongVersionEnd %}

\LongVersion %{
%%%%%%%%%%%%%%%%%%%%%%%%%%%%%%%%%%%%%%%%%%%%%%%%%%%%%%%%%%%%%%%%%%%%%%%%%%%%%%
\section{PoA of the GND Game: Upper Bound and Lower Bound}
\label{sec:poa-gnd-game}
%%%%%%%%%%%%%%%%%%%%%%%%%%%%%%%%%%%%%%%%%%%%%%%%%%%%%%%%%%%%%%%%%%%%%%%%%%%%%%
A byproduct obtained in this paper is a tight bound on the PoA of the GND
games for a class of CSMs.
In \cite{Roughgarden2015IRP}, it is proved that the PoA of a smooth cost
minimization game is
\[
\inf \left\lbrace
\frac{\lambda}{1 - \mu} : \text{ the game is $(\lambda, \mu)$-smooth}
\right\rbrace\, .
\]
Such a PoA is said to be \emph{robust} and can be extended to the mixed Nash
equilibrium, correlated equilibrium, and coarse correlated equilibrium
\cite{Roughgarden2015IRP}.
From Theorem
\ref{theorem_smoothness_of_polynomial_bounded_cost_sharing_mechanisms}, it can
be inferred that:

\begin{theorem}\label{theorem_robust_poa_of_polynomial_bounded_cost_sharing_mechanism}
For any REP-expanded CSM $M$, the induced GND game has a \emph{robust} PoA of $\displaystyle O\Left(\max_{e \in E} \min_{j \in [q]}\Big( \frac{\sigma_{e}}{\xi_{e,\, j}} \Big)^{\frac{1}{\alpha_{j}}}\Right)$.
\end{theorem}

In the following, we prove that the upper bound in Theorem
\ref{theorem_robust_poa_of_polynomial_bounded_cost_sharing_mechanism}
asymptotically matches the lower bound on the PoA of the games induced by a
wide class of CSMs, the \emph{budget-balanced} ones.

\begin{definition}[Budget-balanced cost sharing\cite{Chen2010DNP, Von2013OCS}]\label{definition_of_budget_balanced_cost_sharing_mechanism}
A CSM $M$ is budget-balanced if for any player $i$ and any edge $e$:
\begin{equation}
\sum_{i \in S_{e}} f_{i, e}(p) = F_{e}(l_{e}^{p}) \, . \label{formula_budget_balance}
\end{equation}
\end{definition}

\begin{theorem}\label{theorem_lower_bound_on_poa_of_budget_balanced_cost_sharing_mechanism}
For any budget-balanced CSM, there exists induced GND games with a PoA of $\Omega\Left( \max_{e}\min_{j}\Big( \frac{\sigma_{e}}{\xi_{e,\, j}} \Big)^{\frac{1}{\alpha_{j}}} \Right)$.
\end{theorem}

\begin{figure}[t]
\center
\includegraphics[scale=0.6]{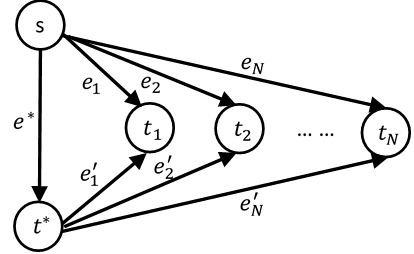}
\caption{\label{figure_directed_simple_network}%
A network for proving the lower bound on the PoA.}
\end{figure}

\begin{proof}
Consider the GND game induced by the GND problem with routing requests and the
budget-balanced CSM.
Suppose that the number of players $N = \Big( \frac{\sigma}{\xi} \Big)^{\frac{1}{\alpha}}$, where $\sigma > 0$ and $\xi > 0$ are properly chosen such that $\frac{\sigma}{\xi}$ is large enough, and $\alpha > 1$. Construct a graph as shown in Fig.~\ref{figure_directed_simple_network}. For every $i \in [N]$, there are an edge $e_{i}$ from $s$ to $t_{i}$, and an edge $e_{i}'$ from $t^{*}$ to $t_{i}$. The nodes $s$ and $t^{*}$ are connected by $e^{*}$. The parameters in the energy cost functions are set as follows:
\begin{itemize}
	\item $\sigma_{e^{*}} = \frac{N}{N + 1} \cdot \sigma$, and $\xi_{e^{*}, 1} = \frac{N}{N + 1} \cdot \xi$.
	\item For every $i \in [N]$, $\sigma_{e_{i}} = \sigma$, and $\xi_{e_{i}, 1} = \xi$.
	\item For every $i \in [N]$, $\sigma_{e_{i}'} = \frac{1}{N + 1} \cdot \sigma$, and $\xi_{e_{i}', 1} = \frac{3}{N + 1} \cdot \xi$.
	\item $\alpha_{1} = \alpha$.
	\item For every $2 \leq j \leq q$, $1 < \alpha_{j} < \alpha_{1}$, $\xi_{e,\, j} < \xi_{e,\, 1} / [q N^{\alpha_{j}}(N + 1)]$.
\end{itemize}
With these settings, we have $\max_{e}\min_{j}\Big(\frac{\sigma_{e}}{\xi_{e,\, j}} \Big)^{\frac{1}{\alpha_{j}}} = \Big(\frac{\sigma}{\xi}\Big)^{\frac{1}{\alpha}}$.

For each player $i$, its source and target are set to $(s, t_{i})$, and its weight $w_{i}(e)$ is set to $1$ for every edge $e$. Then each player $i$ has two alternative paths, $\lbrace e_{i} \rbrace$ and $\lbrace e^{*}, e_{i}' \rbrace$. Consider the profile where every player $i$ chooses the path $\lbrace e_{i} \rbrace$. With any budget-balanced CSM, the cost of each player $i$ must be:
\begin{equation}
f_{i, e_{i}}(p) \; = \; F_{i, e_{i}}(1) \; = \; \sigma_{e_{i}} + \sum_{j = 1}^{q}\xi_{e_{i},\, j} \; < \; \sigma + \xi \cdot \frac{N + 2}{N + 1} \nonumber
\end{equation}
The first transition follows from Definition \ref{definition_of_budget_balanced_cost_sharing_mechanism}, since $S_{e_{i}}$ only contains player $i$ for every $i \in [N]$. The last transition holds since $N^{\alpha_{j}} > 1$ for every $j$. If any player $i$ changes its choice to $\lbrace e^{*}, e_{i}' \rbrace$, with any budget-balanced CSM, its cost would be:
\begin{align}
f_{i, e^{*}}(p) + f_{i, e_{i}'}(p) \; =& \; F_{i, e^{*}}(1) + F_{i, e_{i}'}(1) \nonumber\\
	>&\; \sigma_{e^{*}} + \xi_{e^{*},\, 1} + \sigma_{e_{i}'} + \xi_{e_{i}',\, 1} \nonumber\\ 
	=&\;  \frac{N}{N + 1} \sigma + \frac{N}{N + 1} \xi + \frac{1}{N + 1} \sigma + \frac{3}{N + 1} \xi \nonumber\\ 
	=&\; \sigma + \frac{N + 3}{N + 1}\xi \nonumber%\\
	%\geq&\; \sigma + \frac{N + 1}{N}\xi \nonumber
\end{align}
where the first equality still follows Definition \ref{definition_of_budget_balanced_cost_sharing_mechanism}. Therefore, any player $i$ cannot decrease its cost through a unilateral deviation. By definition, this profile is a pure Nash equilibrium.

The total cost incurred by this equilibrium is $N \cdot (\sigma_{e_{i}} + \sum_{j}\xi_{e_{i},\, j}) > N(\sigma + \xi)$. In contrast, if every player chooses the path $\lbrace e^{*}, e_{i}' \rbrace$, the total cost would be:
\begin{align}
\sigma_{e^{*}} + \sum_{j}\xi_{e^{*},\, j} \cdot N^{\alpha_{j}} + N \cdot
(\sigma_{e_{i}'} + \sum_{j}\xi_{e_{i}',\, j})
\, < \, &
\frac{N}{N + 1} c + \frac{N}{N + 1} \xi \cdot N^{\alpha} \\
& +
N(\frac{1}{N + 1}
\sigma + \frac{3}{N + 1} \xi) + \frac{2}{N + 1}\xi \nonumber \\
= \, &
\frac{N}{N + 1} \sigma + \frac{N}{N + 1} \sigma + \frac{N}{N + 1} \sigma +
\frac{3N + 2}{N + 1} \xi \nonumber \\
< \; &
3 (\sigma + \xi) \nonumber
\end{align}
Thus, the price of anarchy would be at least $\frac{N(\sigma + \xi)}{3(\sigma + \xi)} = \frac{N}{3}$. Since $N = (\frac{\sigma}{\xi})^{\frac{1}{\alpha}}$, this theorem follows.
\end{proof}

From Lemma \ref{lemma_fair_cost_sharing_is_polynomial_bounded}, Lemma
\ref{lemma_shapley_cost_sharing_is_polynomial_bounded}, and Theorem
\ref{theorem_robust_poa_of_polynomial_bounded_cost_sharing_mechanism}, it
follows that for both the proportional fair CSM and the Shapley CSM, the
induced GND games have a PoA of $O\Big( \max_{e}\min_{j}(\sigma_{e}/\xi_{e,\,
j})^{1/\alpha_{j}} \Big)$.
Then from Theorem
\ref{theorem_lower_bound_on_poa_of_budget_balanced_cost_sharing_mechanism}, we
know that these two natural CSMs are asymptotically optimal in the class of
budget balanced CSMs, since they trivially follow
Eq.~(\ref{formula_budget_balance}).
\LongVersionEnd %}

\LongVersion %{
%%%%%%%%%%%%%%%%%%%%%%%%%%%%%%%%%%%%%%%%%%%%%%%%%%%%%%%%%%%%%%%%%%%%%%%%%%%%%%
\section{Alternative Approaches}
\label{sec:altern-appr}
%%%%%%%%%%%%%%%%%%%%%%%%%%%%%%%%%%%%%%%%%%%%%%%%%%%%%%%%%%%%%%%%%%%%%%%%%%%%%%

%%%%%%%%%%%%%%%%%%%%%%%%%%%%%%%%%%%%%%%
\subsection{Learning Based Algorithm for the GND problem with Routing Requests}
\label{subsection_learning_based_algorithm_for_EER}
%%%%%%%%%%%%%%%%%%%%%%%%%%%%%%%%%%%%%%%
Up to now, a set of game theoretic results, such as the smoothness parameters, have been established to investigate the performance of \texttt{Alg-ABRD}. Nevertheless, \texttt{Alg-ABRD} is not the only framework that can utilize the smoothness to generate outputs with desired approximation ratio. In this part, we start to introduce a learning-based technique \cite{Roughgarden2015IRP}, which can also guarantee a good approximation for the GND problem with routing requests when the optimal cost $C^{*}$ of the input instance has a constant lower bound.

\begin{definition*}[GND problem with Routing Requests]
In the GND problem with routing requests, the resources are represented by the set $E$ of edges in a given graph $G = (V, E)$, where $V$ is the set of nodes. The feasible reply collection $P_{i}$ of each requirement $i$ is composed of the paths which connect the associated source-target node pair, and contain no repeating edges.
\end{definition*} 

\begin{definition*}[Problem of Online Decision\cite{Kalai2005EAO}]
Consider an online problem where the input consists of a network $G = (V, E)$ and a sequence of $T'$ cost vectors $\Big\{ \tau^{t} = \{ \tau^{t}(e) \}_{e \in E} \Big\}_{t \in [T']}$, where $\tau^{t}(e) \in [0, 1]$. For each $t \in [T']$, this online problem requires a path $r^{t}$ between a given source-target node pair without any knowledge of the cost vectors $\{ \tau^{t}, \tau^{t + 1}, \cdots, \tau^{T'} \}$. The objective is to minimize the $\mathtt{REGRET}$, which is defined as follows: %\footnote{Here we use the symbol $\max_{\{ \mathtt{COST}^{t} \}_{t \in [T']}}$ because the input cost vectors are assumed to be generated by an adversary.},
\begin{equation*}
\mathtt{REGRET} \; = \;  \sum_{t = 1}^{T'}\sum_{e \in r^{t}} \tau^{t}(e) - \min_{r': r' \text{ connects } \{s, t\}} \sum_{t = 1}^{T'}\sum_{e \in p'} \tau^{t}(e) 
\end{equation*}
\end{definition*}

\begin{lemma}[Follow the Perturbed Leader (FPL) \cite{Kalai2005EAO}]\label{lemma_follow_the_perturbed_leader}
For the problem of online decision, there exists a randomized online learning
algorithm called FPL \cite{Kalai2005EAO} that can compute every $r^{t}$ in
$O(|E| + |V|\log|V|)$-time such that the expected value of \texttt{REGRET} is
no larger than $2|V|\sqrt{|E|T'}$.
\end{lemma}

Using FPL as a subroutine, a learning based algorithm, referred to as
\texttt{L-APX}, is constructed as follows for the GND problem with routing
requests.
The first step is to transform the given problem instance $\mathcal{I}$ to an
GND game by employing the proportional fair CSM, and divide every $\sigma_{e}$ and every $\xi_{e, j}$ by a large enough number such that the cost share of any player on any edge is in the interval $[0, 1]$.
Obviously, such a linear scaling on the cost functions $\lbrace F_{e} \rbrace$
does not influence the approximation ratio.
Then, generate $T' = 4N^{2}|V|^{2}|E|$ strategy profiles $\lbrace \bar{p}^{t}
\rbrace_{t \in [T']}$.
For every $t$ and every player $i$, the path $\bar{p}_{i}^{t}$ is obtained by
running FPL with $\tau_{i}^{t'} = \Big\lbrace \lbrace \tau_{i}^{t'}(e) = f_{i, e}(S_{e}^{t'} \cup \lbrace i \rbrace) \rbrace_{e \in E} \Big\rbrace_{t' \in [t - 1]}$ as the input.
Note that with the proportional fair CSM, the \emph{exact}
cost share of each player on each edge can be obtained in constant time.
Finally, choose one strategy profile $t^{*}$ from $[T']$ randomly and uniformly, and output $\bar{p}^{t^{*}}$.

\begin{lemma}\label{lemma_time_complexity_of_online_learning_based_algorithm}
The algorithm \texttt{L-APX} has a time complexity of $O(N^{3}|V|^{2}|E|^{2}\log|V|)$.
\end{lemma}

\begin{theorem}\label{theorem_approximation_ratio_of_online_learning_based_algorithm}
Let $C'^{*}$ be the optimal result with respect to the linearly scaled cost
functions.
If $C'^{*}$ is known to have a constant lower bound $\mathtt{LB}$, then algorithm \texttt{L-APX} guarantees an approximation ratio of $O\bigg( \max_{e \in E}\min_{j \in [q]}\Big( \frac{\sigma_{e}}{\xi_{e, j}} \Big)^{\frac{1}{\alpha_{j}}} \bigg)$ for the GND problem with routing requests.
\end{theorem}

\begin{proof}
According to \cite{Roughgarden2013NRD}, for a $(\lambda, \mu)$-smooth game, by generating the strategy profiles $\lbrace \bar{p}^{t} \rbrace_{t \in [T']}$ through a randomized algorithm for the problem of online decision, the chosen strategy profile $\bar{p}^{t^{*}}$ guarantees that:
\begin{equation*}
\mathbb{E}_{\bar{p}^{t^{*}}}[C'(\bar{p}^{t^{*}})] \; \leq \; \frac{\lambda}{1 - \mu} C'^{*} + \frac{1}{1 - \mu}\frac{\sum_{i = 1}^{N}\mathbb{E}_{\lbrace \bar{p}_{i}^{t} \rbrace_{t \in [T']}}[\mathtt{REGRET}]}{T'}
\end{equation*}
where $C'(\bar{p}^{t^{*}})$ represents the total cost with respect to the
scaled cost functions.
By Theorem
\ref{theorem_smoothness_of_polynomial_bounded_cost_sharing_mechanisms}, the
values of $\gamma_{\alpha}$ and $\lambda_{\alpha}$ are not influenced by the
linear scaling on $c_{e}$ and $\xi_{e, j}$.
Thus,
\begin{align*}
\mathbb{E}_{\bar{p}^{t^{*}}} [C'(\bar{p}^{t^{*}})]
\, \leq & \,
2 (\gamma_{\alpha} + \lambda_{\alpha}) C'^{*} +
\frac{2 \cdot \sum_{i = 1}^{N}\mathbb{E}_{\lbrace
\bar{p}_{i}^{t} \rbrace_{t \in [T']}}[\mathtt{REGRET}]}{T'} \\
\leq & \,
2 (\gamma_{\alpha} + \lambda_{\alpha}) C'^{*} +
\frac{N \cdot 4 |V| \sqrt{|E| T'}}{T'} \\
\leq & \,
2 \left( \gamma_{\alpha} + \lambda_{\alpha} + \frac{1}{\mathtt{LB}} \right) C'^{*}
\end{align*}
The second line above follows from Lemma \ref{lemma_follow_the_perturbed_leader}. The third line holds because it is assumed that $\mathtt{OPT'} \geq \mathtt{LB}$. Since $\mathtt{LB}$ is a constant, this theorem holds.
\end{proof}

Lemma \ref{lemma_time_complexity_of_online_learning_based_algorithm} and Theorem \ref{theorem_approximation_ratio_of_online_learning_based_algorithm} indicate that, \texttt{L-APX} promises the same upper bound on the approximation ratio as \texttt{Alg-ABRD} for the special input instances with $C^{'*} \geq \mathtt{LB}$, and when the given graph has a small size while the number of requirements is large, \texttt{L-APX} has a better time complexity than \texttt{Alg-ABRD}. However, it remains unknown for us that how to generalize Theorem \ref{theorem_approximation_ratio_of_online_learning_based_algorithm} to the general case where there is no guarantee for the lower bound of the optimal solution. The critical issue here is that even before the linear scaling, the optimal result $C^{*}$ can be arbitrarily small. This problem is left for future research.

%%%%%%%%%%%%%%%%%%%%%%%%%%%%%%%%%%%%%%%
\subsection{Convex Programming and Rounding for the GND Problem with Routing Requests}
\label{subsection_convex_programming_and_rounding}
%%%%%%%%%%%%%%%%%%%%%%%%%%%%%%%%%%%%%%%
The approach presented in this subsection was suggested to us by an anonymous reviewer for a special case of the GND problem.
Specifically, like the approach presented in
\Sect{}~\ref{subsection_learning_based_algorithm_for_EER}, this approach also
addresses the GND problem with routing requests, but restricts the attention
to the more specific case where the given graph is undirected and the weights
of the requirements are related.
(Recall that related weights means that the weight of every requirement $i$
satisfies $w_{i}(e) = w_{i}$ for every $e \in E$.)
Furthermore, in this part, the cost function $F_{e}$ of each edge $e$ is
assumed to be an energy consumption cost function
(\ref{equation:energy-consumption-cost-function}), a specific (simpler)
form of the REP cost functions (\ref{equation:REP-cost-function}) used in the
other parts of this paper.
We shall refer to this specific GND restriction as \emph{energy efficient
routing (EER)}.

Since this approach is based on convex programming and rounding, we shall
refer to the resulting approximation algorithm as \texttt{CPR}.
In this algorithm, every requirement
$i \in [N]$
is first partitioned into a set
$\mathcal{R}_{i}'$ of $w_{i}$ \emph{sub-requests}, where every sub-request $i_{j}$ is associated with the same source-target node pair as the original requirement $i$, and has the same weight $w_{i_{j}} = 1$.
Such a partition is feasible since the weights are assumed to be related.
Let $\mathcal{R}' = \bigcup_{i}\mathcal{R}_{i}'$ be the set of all the sub-requests, $\mathcal{I}'$ be the instance obtained by replacing the set of requirements in the given EER instance $\mathcal{I}$ with $\mathcal{R}'$, and $\widehat{\mathcal{I}}'$ be a variant instance which replaces the energy consumption cost function $F_{e}$ in $\mathcal{I}'$ with the following variant cost function:
\begin{equation}
\widehat{F}_{e}(p') \; = \;
\begin{cases}
0 &  l_{e}^{p'} = 0 \\
\sigma_{e} + \xi_{e} \cdot \left( (l_{e}^{p'})^{\alpha} + \sum_{i}(w_{i})^{\alpha-1} \cdot l_{e}^{p'}(i) \right)  & l_{e}^{p'} > 0
\end{cases} \, , \label{formula_revised_cost_function}
\end{equation}
where $p'$ is an arbitrary feasible path profile for $\mathcal{R}'$, $l_{e}^{p'}(i)$  is the load incurred by routing sub-requests in $\mathcal{R}_{i}'$ along $e$. 
Let $\widehat{p}^{*}$ be the \emph{fractional} optimal solution of $\widehat{\mathcal{I}}'$. Similar with \cite{AzarE2005}, it can be proved that:

\begin{lemma}\label{lemma_gap_between_optimal_solutions_of_original_instance_and_revised_one}
$\widehat{C}(\widehat{p}^{*}) < 2 \cdot C^{*}$, where $\widehat{C}(\cdot)$ denotes the total cost with respect to Eq.~\eqref{formula_revised_cost_function}.
\end{lemma}

\begin{proof}
The path profile $p^{*}$ also induces a feasible solution for $\widehat{\mathcal{I}}'$, which routes $i_{j}$ along the path $p_{i}^{*}$ for every $i \in [N]$. Let the total cost incurred by this feasible solution be $\widehat{C}(p^{*})$, then:
\begin{align}
\widehat{C}(\widehat{p}^{*}) \; \leq \; \widehat{C}(p^{*}) \; =& \; \sum_{e \in p^{*}}\bigg[\sigma_{e} + \xi_{e} \cdot \Big( (l_{e}^{p^{*}})^{\alpha} + \sum_{i}(w_{i})^{\alpha-1} \cdot l_{e}^{p_{i}^{*}} \Big) \bigg] \nonumber\\
=& \; \sum_{e \in p^{*}}\bigg[\sigma_{e} + \xi_{e} \cdot \Big( (l_{e}^{p^{*}})^{\alpha} + \sum_{i \in p^{*}}(w_{i})^{\alpha} \Big) \bigg] \nonumber \\
\leq& \; \sum_{e \in p^{*}}\left[\sigma_{e} + \xi_{e} \cdot \Big( (l_{e}^{p^{*}})^{\alpha} +  (l_{e}^{p^{*}})^{\alpha} \Big) \right] \nonumber \\
<& \; 2 \cdot  \sum_{e \in p^{*}}\Big( \sigma_{e} + \xi_{e} \cdot (l_{e}^{p^{*}})^{\alpha} \Big) \nonumber\\
=& \; 2 \cdot C^{*} \nonumber
\end{align}
The second line holds since $p^{*}$ is an integral solution of $\mathcal{I}$ \cite{AzarE2005}. The third line following the fact $\sum_{i \in p^{*}}w_{i} = l_{e}^{p^{*}}$ and the superadditivity of the power function. The fourth line holds since $\sigma_{e} > 0$.
\end{proof}

The next step of \texttt{CPR} is to utilize the convex programming based technique proposed in \cite{Andrews2012RPM} to generate a solution for the instance $\widehat{\mathcal{I}}'$. In particular, it converts $\widehat{F}_{e}(p')$ to a convex cost function $\bar{F}_{e}(p') = W_{e}(p') + \xi_{e}\sum_{i}(w_{i})^{\alpha-1} \cdot l_{e}^{p'}(i)$, where:
\begin{equation*}
W_{e}(p') \; = \;
\begin{cases}
\zeta_{e} \cdot (l_{e}^{p'}) &  l_{e}^{p'} \in \left[0, \max\Big\lbrace 1, \Big( \frac{\sigma_{e}}{(\alpha - 1)\mu_{e}} \Big)^{\frac{1}{\alpha}} \Big\rbrace\right] \\
\sigma_{e} + \xi_{e} \cdot (l_{e}^{p'})^{\alpha}  & l_{e}^{p'} > \max\Big\lbrace 1, \Big( \frac{\sigma_{e}}{(\alpha - 1)\mu_{e}} \Big)^{\frac{1}{\alpha}} \Big\rbrace
\end{cases}
\, , 
\end{equation*}
and
\begin{equation*}
\zeta_{e} \; = \;
\begin{cases}
\sigma_{e} + \mu_{e} & \Big( \frac{\sigma_{e}}{(\alpha - 1)\mu_{e}} \Big)^{\frac{1}{\alpha}} < 1 \\
\alpha\xi_{e}\Big(\frac{\sigma_{e}}{(\alpha - 1)\xi_{e}}\Big)^{1 - \frac{1}{\alpha}} & \Big( \frac{\sigma_{e}}{(\alpha - 1)\mu_{e}} \Big)^{\frac{1}{\alpha}} < 1 \geq 1
\end{cases}
\, .
\end{equation*}
Following \cite{Andrews2012RPM}, \texttt{CPR} solves the convex programming problem with respect to the convex cost function $\bar{F}_{e}(p')$ to obtain a fractional solution $\widehat{p}^{\circ}$, then rounds it to an integral solution $\widehat{p}^{\sharp}$ in a randomized manner. Using the linearity of expectation, for every edge $e \in E$:
\begin{equation}
\mathbb{E}\left[ \xi_{e}\sum_{i}(w_{i})^{\alpha-1} \cdot l_{e}^{\widehat{p}^{\sharp}}(i) \right] \; = \; \xi_{e}\sum_{i}(w_{i})^{\alpha-1} \cdot \mathbb{E}\left[ l_{e}^{\widehat{p}^{\sharp}}(i) \right] \; = \; \xi_{e}\sum_{i}(w_{i})^{\alpha-1} \cdot l_{e}^{\widehat{p}^{\circ}}(i) \label{formula_linear_part_after_rounding_the_fractional_solution}
\end{equation}
Furthermore, it is proved in \cite{Andrews2012RPM} that the random rounding technique ensures that:
\begin{equation}
\mathbb{E}\left[ W_{e}(\widehat{p}^{\sharp}) \right] \; \leq \; O\left( \Big( \frac{\sigma_{e}}{\xi_{e}} \Big)^{\frac{1}{\alpha}} \right) \cdot W_{e}(\widehat{p}^{\circ}) \label{formula_convex_part_after_rounding_the_fractional_solution}
\end{equation}

Combining Eq.~\eqref{formula_linear_part_after_rounding_the_fractional_solution} with Eq.~\eqref{formula_convex_part_after_rounding_the_fractional_solution}, we have
\begin{equation}
\sum_{e} \mathbb{E}\left[ \bar{F}_{e}(\widehat{p}^{\sharp}) \right] \; \leq \; O\left( \Big( \max_{e} \frac{\sigma_{e}}{\xi_{e}} \Big)^{\frac{1}{\alpha}} \right) \cdot \sum_{e} \bar{F}_{e}(\widehat{p}^{\circ}) \label{formula_gap_between_the_first_rounding_solution_and_the_fractional_optimal_given_by_convex_programming}
\end{equation}

Since $\widehat{p}^{\circ}$ is the optimal fractional solution obtained by the convex programming with respect to the cost function $\bar{F}_{e}(p')$, and $\bar{F}_{e}(p') \leq \widehat{F}_{e}(p')$ for any edge $e$ and any profile $p'$ that is feasible for $\mathcal{R}'$ \cite{Andrews2012RPM}:
\begin{equation}
\sum_{e} \bar{F}_{e}(\widehat{p}^{\circ}) \; \leq \; \sum_{e} \bar{F}_{e}(\widehat{p}^{*}) \; \leq \; \sum_{e} \widehat{F}_{e}(\widehat{p}^{*}) \; = \; \widehat{C}(\widehat{p}^{*}) \label{formula_relationship_between_two_different_optimal_fractional_solution}
\end{equation}

Recall that $\widehat{C}(\widehat{p}^{*})$ is the optimal fractional solution of $\widehat{\mathcal{I}}'$, Eq.~\eqref{formula_gap_between_the_first_rounding_solution_and_the_fractional_optimal_given_by_convex_programming} and Eq.~\eqref{formula_relationship_between_two_different_optimal_fractional_solution} imply that:

\begin{lemma}\label{lemma_approximation_for_the_revised_instance_with_sub_requests}
The solution $\widehat{p}^{\sharp}$ is an $O\left( \Big( \max_{e} \frac{\sigma_{e}}{\xi_{e}} \Big)^{\frac{1}{\alpha}} \right)$-approximation solution of the instance $\widehat{\mathcal{I}}$.
\end{lemma}

The last step of \texttt{CPR} is to convert $\widehat{p}^{\sharp}$ to a integral solution $p^{\sharp}$ that is feasible for the original instance $\mathcal{I}$, still by random rounding. Particularly, to generate $p^{\sharp}$, each traffic request $i$ in instance $\mathcal{I}$ should be routed along the path $\widehat{p}_{i_{j}}^{\sharp} \in \widehat{p}^{\sharp}$ with probability $\displaystyle \frac{1}{w_{i}}$. Then we have
\begin{align}
\mathbb{E}\left[\sum_{e} F_{e}({p}^{\sharp}) \right]  \;\leq& \; \mathbb{E}\left[ \sum_{e \in \widehat{p}^{\sharp}}\left( \sigma_{e} +  \xi_{e}  (l_{e}^{p^{\sharp}})^{\alpha}  \right) \right] \nonumber \\
=& \;  \mathbb{E}\left[ \sum_{e \in \widehat{p}^{\sharp}}\sigma_{e} \right] + \mathbb{E}\left[ \sum_{e \in \widehat{p}^{\sharp}}\xi_{e}  (l_{e}^{p^{\sharp}})^{\alpha}  \right] \nonumber\\
\leq&\; \mathbb{E}\left[ \sum_{e \in \widehat{p}^{\sharp}}\sigma_{e} \right] +
O(\alpha^{\alpha}) \mathbb{E}\left[ \sum_{e \in \widehat{p}^{\sharp}} \xi_{e} (l_{e}^{\widehat{p}^{\sharp}})^{\alpha} + \sum_{i}(w_{i})^{\alpha - 1}l_{e}^{\widehat{p}^{\sharp}}(i) \right]
\nonumber \\
\leq&\; O(\alpha^\alpha)\mathbb{E}\left[\sum_{e \in \widehat{p}^{\sharp}} \widehat{F}_{e}(\widehat{p}^{\sharp}) \right]\label{formula_approximation_factor_after_second_rounding}
\end{align}
The third transition above follows from \cite{Anupam2017OPD}. Combining Lemma \ref{lemma_gap_between_optimal_solutions_of_original_instance_and_revised_one}, Lemma \ref{lemma_approximation_for_the_revised_instance_with_sub_requests}, and Eq.~\eqref{formula_approximation_factor_after_second_rounding} gives the following result:

\begin{theorem}\label{theorem_final_approximation_ratio_of_AREA}
The algorithm \texttt{CPR} has an approximation ratio of $O\left( \Big( \max_{e} \frac{\sigma_{e}}{\xi_{e}} \Big)^{\frac{1}{\alpha}} \right)$.
\end{theorem}

\paragraph{Weight-scaling and Loss in Approximation Ratio}

The algorithm \texttt{CPR} should process every sub-request independently, therefore, its time complexity depends on the numeric value $\sum_{i}w_{i}$, which cannot be bounded by a polynomial of the instance size, $\operatorname{poly}(\mathcal{I})$. For this issue, a native idea is to scale down and round the weights so that $w_{i}$ is bounded by a polynomial of $|\mathcal{I}|$. Although this \emph{weight-scaling} technique works well for some classical optimization problem, such as the Knapsack problem, it is possible to incur a significant loss in the approximation ratio for our problem.

\begin{definition}[Weight-scaling Function]\label{definition_weight_scaling_function}
Given an EER instance $\mathcal{I}$, a weight-scaling function $\mathtt{WSF}$ maps the weight vector $\lbrace w_{i} \rbrace_{i \in [N]}$ of $\mathcal{I}$ to $\lbrace \mathtt{WSF}_{\mathcal{I}}(w_{i}) \rbrace_{i \in [N]} \in (\mathbb{Z}_{> 0})^N$ such that $\max_{i \in [N]} \mathtt{WSF}_{\mathcal{I}}(w_{i})$ is bounded by $\operatorname{poly}(|\mathcal{I}|)$.
\end{definition}

\begin{definition}[Ambiguity]
A weight-scaling function is said to be $\sigma$-ambiguous for a positive number $\sigma > 1$ if there exists an instance $\mathcal{I}$ such that $\mathtt{WSF}_{\mathcal{I}}(w_{i}) = \mathtt{WSF}_{\mathcal{I}}(w_{i'})$ for every $w_{i} = 1$ and $w_{i'} = \sigma^{\frac{1}{\alpha}}$.
\end{definition}

To illustrate the ambiguity, consider the weight scaling function in \cite{Panigrahi2015DAA}, which maps each weight $w_{i}$ to $\displaystyle \left\lceil \frac{w_{i} \cdot N}{\varepsilon \max_{i}w_{i}} \right\rceil \cdot \frac{\varepsilon \max_{i}w_{i}}{N}$ where $\varepsilon$ is a positive constant less than $1$. Such a weight-scaling function is $\sigma$-ambiguous for any $\sigma > 1$ when there exists a request $i_{L}$ with a large enough weight $\displaystyle w_{i_{L}} \geq \frac{N\sigma^{\frac{1}{\alpha}}}{\varepsilon}$. Such kind of high-weight request is ignored in the following part, as it can be placed on an independent subgraph where the edges have sufficiently small parameters $\sigma_{e}$ and $\xi_{e}$ such that the existence of $i_{L}$ does not influence the analysis.

To analyze the loss in approximation ratio caused by a $\sigma$-ambiguous weight-scaling function, consider a graph $G_{\sigma}$ constructed as follows. There are two nodes in $G$, $s$ and $t$, which are connected by more than $\sigma^{\frac{1}{2\alpha - 1}}$ edges. There exists a special edge $e^{*}$ with cost parameters $\sigma_{e^{*}} = 1$ and $\xi_{e^{*}} = 1$, while for any edge $e \in E - \lbrace e^{*} \rbrace$, $\sigma_{e} = \sigma$ and $\xi_{e} = 1$. Following theorem shows that on this graph, a $\sigma$-ambiguous weight-scaling function leads to an approximation ratio that is definitely larger than the ratio $O\Big(\max_{e}\Big(\frac{\sigma_{e}}{\xi_{e}}\Big)^{\frac{1}{\alpha}}\Big) = O\Big( \sigma^{\frac{1}{\alpha}}\Big)$ when $\alpha > \frac{3 + \sqrt{5}}{2}$.

\begin{theorem}\label{theorem_loss_in_approximation_ratio_caused_by_ambiguous_weight_scaling}
After employing a $\sigma$-ambiguous weight-scaling function, the approximation ratio of any deterministic algorithm for the EER instances on the graph $G_{\sigma}$ is $\Omega\left( \sigma^{\frac{\alpha - 1}{2\alpha - 1}} \right)$.
\end{theorem}

\begin{proof}
Consider two input instances, $\mathcal{I}_{1}$ and $\mathcal{I}_{2}$, on the graph $G_{\sigma}$. Each of these two instances contains $N = \sigma^{\frac{1}{2\alpha - 1}}$ requests with the source-target pair $\lbrace s, t\rbrace$. For every $i \in [N]$, the weight of player $i$ is respectively set to $1$ and $\sigma^{\frac{1}{\alpha}}$ in $\mathcal{I}_{1}$ and $\mathcal{I}_{2}$. Let $\mathcal{I}_{1}^{\mathtt{MST}}$ and $\mathcal{I}_{2}^{\mathtt{MST}}$ be the corresponding instances processed by a $\sigma$-ambiguous weight-scaling function. By the definition, no algorithm can distinguish $\mathcal{I}_{1}^{\mathtt{MST}}$ from $\mathcal{I}_{2}^{\mathtt{MST}}$. Hence, the following observation trivially holds:

\begin{claim}\label{observation_same_profile_for_ambiguous_instances_by_deterministic_algorithm}
For any deterministic algorithm for the EER instance, the output generated for $\mathcal{I}_{1}^{\mathtt{MST}}$ is same as the one for $\mathcal{I}_{2}^{\mathtt{MST}}$.
\end{claim} 

Let the output generated by a given deterministic algorithm of the EER problem for $\mathcal{I}_{1}^{\mathtt{MST}}$ be $p$. Denote the optimal solution of $\mathcal{I}_{1}$ (resp.~$\mathcal{I}_{2}$) by $C^{*}(\mathcal{I}_{1})$ (resp.~$C^{*}(\mathcal{I}_{2})$), and the total cost incurred by $p$ for $\mathcal{I}_{1}$ (resp.~$\mathcal{I}_{2}$) be $C(\mathcal{I}_{1}, p)$ (resp.~$C(\mathcal{I}_{2}, p)$).

\sloppy
\begin{claim}\label{observation_approximation_caused_by_obliviousness_to_weights}
The approximation ratio of the given deterministic algorithm is at least $\max\Big\{ \frac{C(\mathcal{I}_{1}, p)}{C^{*}(\mathcal{I}_{1})},  \frac{C(\mathcal{I}_{2}, p)}{C^{*}(\mathcal{I}_{2})} \Big\}$.
\end{claim}
\par\fussy

\begin{claim}\label{observation_balance_between_the_ratios_for_significantly_separated_instances}
For any profile $p$, $\max\Big\{ \frac{C(\mathcal{I}_{1}, p)}{C^{*}(\mathcal{I}_{1})},  \frac{C(\mathcal{I}_{2}, p)}{C^{*}(\mathcal{I}_{2})} \Big\} \geq \frac{1}{2}\sigma^{\frac{\alpha - 1}{2\alpha - 1}}$.
\end{claim}

\begin{subproof}
\qedlabel{observation_balance_between_the_ratios_for_significantly_separated_instances}
Suppose that $x$ edges in $E - \lbrace e^{*} \rbrace$ is used by $p$. Then $C(\mathcal{I}_{1}, p) \geq x \cdot (\sigma + 1)$. Since the cost of routing all requests through $e^{*}$ in the instance $\mathcal{I}$ is $1 + \sigma^{\frac{\alpha}{2\alpha - 1}}$,
\begin{equation*}
\frac{C(\mathcal{I}_{1}, p)}{C^{*}(\mathcal{I}_{1})} \;\geq\; \frac{x(\sigma + 1)}{1 + \sigma^{\frac{1}{\alpha - 1}}} \; \geq \; \frac{x}{2} \cdot \sigma^{\frac{\alpha - 1}{2\alpha - 1}} \;.
\end{equation*}
Noticing that the $x + 1$ edges used by $p$ have the same value of the parameter $\xi_{e}$, we have:
\begin{equation*}
C(\mathcal{I}_{2}, p) \;\geq\; (x + 1) \Big[ \Big( \frac{N}{x + 1} \cdot \sigma^{\frac{1}{\alpha}}\Big)^{\alpha} \Big] \; = \; \sigma^{1 + \frac{\alpha}{2\alpha - 1}} \cdot (x + 1)^{1 - \alpha} \; .
\end{equation*}
By routing each request along a distinct edge, a solution with total cost $2\sigma^{1 + \frac{1}{2\alpha - 1}}$ can be obtained for the instance $\mathcal{I}'$. Therefore,
\begin{equation*}
\frac{C(\mathcal{I}_{2}, p)}{C^{*}(\mathcal{I}_{2})} \; \geq \; \frac{(x + 1)^{1 - \alpha}}{2} \cdot \sigma^{\frac{\alpha - 1}{2\alpha - 1}} \; .
\end{equation*}

Note that $x \in \mathbb{Z}_{\geq 0}$. When $x \geq 1$, $\displaystyle \max\Big\{ \frac{C(\mathcal{I}_{1}, p)}{C^{*}(\mathcal{I}_{1})},  \frac{C(\mathcal{I}_{2}, p)}{C^{*}(\mathcal{I}_{2})} \Big\} \geq \frac{C(\mathcal{I}_{1}, p)}{C^{*}(\mathcal{I}_{1})} \geq \frac{1}{2}\sigma^{\frac{\alpha - 1}{2\alpha -1}}$; while if $x = 0$, $\displaystyle \max\Big\{ \frac{C(\mathcal{I}_{1}, p)}{C^{*}(\mathcal{I}_{1})},  \frac{C(\mathcal{I}_{2}, p)}{C^{*}(\mathcal{I}_{2})} \Big\} \geq \frac{C(\mathcal{I}_{2}, p)}{C^{*}(\mathcal{I}_{2})} \geq \frac{1}{2}\sigma^{\frac{\alpha - 1}{2\alpha -1}}$.
\end{subproof}
This proof is completed by combining the claims above.
\end{proof}

For any random algorithm for the EER instance, it should generate an output for the instance $\mathcal{I}^{1}$ with the same probability distribution over the path profiles as the one for the instance $\mathcal{I}^{2}$. Therefore, it can be verified that:
\begin{corollary}\label{corollary_loss_in_approximation_ratio_of_random_algorithms_with_weight_scaling}
After employing a $\sigma$-ambiguous weight-scaling function, the approximation ratio of any random algorithm for the EER problem instances on the graph $G_{\sigma}$ is $\Omega\left( \sigma^{\frac{\alpha - 1}{2\alpha - 1}} \right)$.
\end{corollary}

As mentioned earlier, the approximation ratio $\Omega\left( \sigma^{\frac{\alpha - 1}{2\alpha - 1}} \right)$ is worse than the approximation ratio in our main result when $\alpha > \frac{3 + \sqrt{5}}{2} \approx 2.618$. Furthermore, by Theorem \ref{theorem_final_approximation_ratio_of_AREA}, we cannot expect an approximation ratio better than our main result when $\alpha \leq 2.618$, either, if we combine \texttt{CPR} with a weight-scaling function.
\LongVersionEnd %}

\LongVersion
\section*{Acknowledgements}
We would like to thank Tim Roughgarden for very useful discussions.
\LongVersionEnd

%%%%%%%%%%%%%%%%%%%%%%%%%%%%%%%%%%%%%%%%%%%%%%%%%%%%%%%%%%%%%%%%%%%%%%%%%%%%%%
%\clearpage
%\begin{figure}[!t]
%\begin{center}
%\LARGE{FIGURES}
%\end{center}
%\end{figure}

%\LongVersion %{

%\LongVersionEnd %}

%%%%%%%%%%%%%%%%%%%%%%%%%%%%%%%%%%%%%%%%%%%%%%%%%%%%%%%%%%%%%%%%%%%%%%%%%%%%%%
%%%%%%%%%%%%%%%%%%%%%%%%%%%%%%%%%%%%%%%%%%%%%%%%%%%%%%%%%%%%%%%%%%%%%%%%%%%%%%
\clearpage
\bibliographystyle{alpha}
\bibliography{references}

\end{document}